\documentclass[11pt,letterpaper]{article}
\usepackage[margin=1in]{geometry}

\usepackage{amsmath, amsthm, amsfonts, amssymb}
\usepackage{dsfont}
\usepackage{autobreak}
\allowdisplaybreaks[4]
\usepackage{booktabs}
\usepackage{xcolor}
\usepackage{hyperref}
\usepackage{url}
\usepackage{appendix}
\usepackage{subfigure}
\usepackage{graphicx}
\usepackage{enumitem}
\usepackage[ruled, linesnumbered]{algorithm2e}
\usepackage[numbers,sort]{natbib}

\usepackage{tikz}
\usetikzlibrary{calc}

\newtheorem{theorem}{Theorem}[section]

\newtheorem{lemma}[theorem]{Lemma}

\newtheorem{definition}[theorem]{Definition}
\newtheorem{fact}[theorem]{Fact}

\newtheorem{remark}[theorem]{Remark}

\newenvironment{proofof}[1]{\begin{proof}[Proof of #1]}{\end{proof}}

\newcommand{\eps}{\varepsilon}
\newcommand{\alg}{\mathrm{ALG}}
\newcommand{\opt}{\mathrm{OPT}}
\renewcommand{\tilde}{\widetilde}
\renewcommand{\emptyset}{\varnothing}
\DeclareMathOperator{\E}{\mathbb{E}}
\newcommand{\pois}{\operatorname{Pois}}
\newcommand{\ber}{\operatorname{Ber}}

\title{On the Cost of Non-Adaptivity in Matroid Prophet Inequalities}
\author{Tianle Jiang\thanks{Supported by NSF grant IIS-2402823.}\\
Duke University\\
\texttt{tianle.jiang@duke.edu}
}
\date{}

\begin{document}

\maketitle

\begin{abstract}
    Matroid prophet inequalities admit an optimal 2-competitive algorithm, which relies on adaptively updating thresholds based on previous outcomes. Motivated by applications to posted-price mechanisms and the structural simplicity of fixed-threshold policies, recent work initiated the study of non-adaptive matroid prophet inequalities. The central question is to understand how much the performance deteriorates when thresholds must be fixed in advance.

    We explore the fundamental limits of non-adaptive algorithms and show new structural barriers and algorithmic insights. We first identify a simple case where non-adaptive algorithms admit a lower bound strictly above $2$: for truncated partition matroids where every local partition has rank $1$, there is an instance giving a lower bound of $\approx 2.179$, and we give a non-adaptive OCRS-style algorithm that exactly matches this ratio. We then show that richer matroid structures can amplify this barrier: we obtain stronger lower bounds of $2.217$ for laminar matroids and $3$ for graphic matroids, and complement the hardness results with improved upper bounds for these matroids.
\end{abstract}

\section{Introduction}
\label{sec:intro}

Prophet inequalities are a fundamental model for online decision-making under uncertainty. 
In the matroid prophet inequality problem, the values of the elements are drawn independently from known distributions and revealed one by one in an online order; upon seeing each realized value, the algorithm must irrevocably decide whether to accept it, subject to the constraint that the accepted elements form an independent set of a given matroid. 
The goal is to compete with the prophet, who knows all realizations in advance and selects the maximum-value feasible set. 

In the classical single-choice setting (i.e., $1$-uniform matroids), the optimal $2$-competitive algorithms of~\citet{Krengel_Sucheston1978,Samuel_Cahn1984} use thresholds that can be fixed before the online process begins. 
For general matroids, \citet{stoc/KleinbergW12} obtained a $2$-competitive algorithm, matching the classical lower bound and therefore achieving the best possible worst-case guarantee. However, the algorithm crucially relies on using adaptive thresholds for each online item that depend on the previous realizations and decisions. This naturally raises the question of how the performance deteriorates if the thresholds must be fixed in advance.

This question is further motivated by the connection between prophet inequalities and posted-price mechanisms.
Prophet inequalities have long been used as a tool for Bayesian mechanism design, starting from the seminal work of~\citet{aaai/HajiaghayiKS07} and later developed more systematically by~\citet{stoc/ChawlaHMS10, focs/DuettingFKL17}, who showed that prophet inequalities can be converted into order-oblivious posted pricing mechanisms for unit-demand buyers. 
However, as emphasized by~\citet{sagt/ChawlaGKM24}, adaptive thresholds are less suitable when one wants to translate a prophet inequality into a truthful mechanism for richer multi-dimensional buyers. 
For example, for a single constrained-additive buyer who may purchase multiple items subject to a matroid constraint, adaptive thresholds correspond to prices that change as a function of the buyer's previous purchases, which can destroy truthfulness. 
Moreover, non-adaptive algorithms also have practical advantages: item prices are anonymous, i.e., for each item the price remains the same for all buyers; the thresholds are computed offline while the online phase requires little additional computation.

Motivated by these considerations, we study the following question: how much does the competitive ratio deteriorate if the thresholds must be fixed before the online process begins? In short, what is the \textbf{cost of non-adaptivity} in matroid prophet inequalities? 
Our understanding of this problem remains limited. \citet{soda/FeldmanSZ16} provide an $\Omega(\log n/\log\log n)$ lower bound for gammoids, implied by their lower bound for online contention resolution schemes. 
\citet{sagt/ChawlaGKM24} initiated the study of non-adaptive algorithms and gave a $32$-competitive non-adaptive algorithm for graphic matroids. 
More recently, \citet{dam/PashkovichS26} obtain $O(1)$-competitive algorithms for several additional matroid classes, and observe that some results can be implied by known single-sample matroid prophet inequalities~\cite{geb/AzarKW19, corr/abs-2103-13089}. 

These results reveal a rich landscape, but a basic problem is left open.
The only known hardness result comes from gammoids, a highly expressive class of matroids that is less commonly studied as a natural matroid object in algorithmic applications.
On the algorithmic side, the known constant guarantees are still far from the adaptive benchmark of $2$, and it remains unclear whether this gap is merely an artifact of the analyses or an inherent limitation of non-adaptive algorithms. 

In this paper, we make progress on understanding the cost of non-adaptivity. For adaptive algorithms, the classical two-variable prophet inequality instance directly implies the tight worst-case ratio of $2$ for many matroids. However, we observe that non-adaptivity introduces new structural barriers and leads to lower bounds strictly larger than $2$ for many natural classes of matroids, including laminar and graphic matroids. We also complement these lower bounds with improved non-adaptive algorithms, providing new algorithmic insights for non-adaptive prophet inequalities.

\subsection{Our Results}

Our results involve new lower and upper bounds for many natural classes of matroids, and are summarized in Table~\ref{tab:overview}. The formal definitions for these matroid classes appear in Section~\ref{sec:matroid_classes}.

\begin{table}[ht]
    \centering
    \caption{Overview of our results and comparison to prior work}
    \label{tab:overview}
    \renewcommand{\arraystretch}{1.18}
    \begin{tabular}{p{0.28\linewidth} p{0.27\linewidth} p{0.36\linewidth}}
        \toprule
        Matroid type & Lower bound & Upper bound \\
        \midrule

        truncated $R_1$-partition
        &
        \(2.179\)~(Theorem~\ref{thm:truncated_lower_bound})
        &
        \(8\)~(\cite{corr/abs-2103-13089})
        \(\to\)
        \(2.179\)~(Theorem~\ref{thm:truncated-R1-partition})
        \\

        truncated partition
        &
        \(2.179\)~(Theorem~\ref{thm:truncated_lower_bound})
        &
        \(8\)~(\cite{corr/abs-2103-13089})
        \(\to\)
        \(3\)~(Theorem~\ref{thm:truncated_upper_bound})
        \\

        laminar
        &
        \(2.217\)~(Theorem~\ref{thm:laminar_lower_bound})
        &
        \(9.6\)~(\cite{geb/AzarKW19})
        \(\to\)
        \(6.311\)~(Theorem~\ref{thm:laminar_upper_bound})
        \\

        graphic
        &
        \(3\)~(Theorem~\ref{thm:graphic_lower_bound_3})
        &
        \(32\)~(\cite{sagt/ChawlaGKM24})
        \(\to\)
        \(8\)~(Theorem~\ref{thm:graphic_upper_bound_8})
        \\

        simple graphic
        &
        \(2\)~(\cite{itcs/AlonGPRWW025})
        &
        \(16\)~(\cite{dam/PashkovichS26})
        \(\to\)
        \(5.828\)~(Theorem~\ref{thm:graphic_upper_bound_simple})
        \\
        
        \bottomrule
    \end{tabular}
\end{table}

\paragraph{Tight analysis of a basic case.}
We begin with a special class of matroids, which we refer to as truncated $R_1$-partition matroids: there is a global cardinality constraint of size $k$, and the items are divided into several disjoint partitions such that an independent set can contain at most one item from each partition. 
This is arguably the simplest natural setting beyond $k$-uniform and partition matroids\footnote{
    For $k$-uniform matroids with $k\geq 2$, an even more restrictive class of algorithms that use a single threshold for all items, namely static threshold algorithms, has been extensively studied \cite{Samuel_Cahn1984, aaai/HajiaghayiKS07, ior/ChawlaDL24, ec/JiangMZ23}, and the competitive ratio is strictly smaller than $2$. This prevents us from showing a separation from $2$ for non-adaptive algorithms using $k$-uniform matroids or partition matroids. In Appendix~\ref{sec:discussion}, we provide a more detailed discussion of static threshold algorithms, and show that they are no longer powerful for richer matroid structures, which justifies our focus on non-adaptive algorithms.
}, and captures realistic scenarios such as selecting $k$ candidates subject to group fairness constraints. In this setting, we show a lower bound of $\approx 2.179$, strictly larger than the adaptive benchmark of $2$, and give a non-adaptive algorithm with exactly the same competitive ratio.

The hard instance starts with many singleton partitions each containing a low-value Bernoulli item, ends with a partition that embeds the classic two-variable hard instance for single-choice prophet inequalities, and has global rank $2$. This captures a fundamental limitation of non-adaptive algorithms: they cannot simultaneously ensure high expected reward from the lower-value Bernoulli items and keep one slot available for the final partition with sufficiently large probability. 

We view the matching upper bound as a main technical contribution of the paper. It is based on an intricate analysis of an OCRS-style process, which rounds the optimal ex-ante solution in an online fashion but computes non-adaptive thresholds in advance so that each item is accepted with a prescribed marginal probability. We characterize the tight competitive ratio by modeling the interaction between the partitions and the global cardinality constraint via a Poisson binomial variable and its mean capped at $k$, and, roughly speaking, demonstrate that the worst case of the analysis has the same structure as the hard instance, which explains why the lower bound and the upper bound coincide.

\paragraph{Stronger separations from richer matroid structure.}
Next we show that the cost of non-adaptivity from the basic case can be amplified by exploiting richer matroid structures. 
For laminar matroids, which allow multiple levels of capacity constraints, we recursively extend the basic construction and obtain a lower bound of $2.217$. 
For graphic matroids, we embed a similar structural barrier into the ``hat graph'' (see also, e.g., \cite{esa/LeeS18, wine/BahraniBSW21}) and obtain a lower bound of $3$.
Note that the construction crucially relies on using parallel edges to implement the two-variable hard instance, and therefore does not extend to simple graphs, where the best known lower bound remains $2$ \cite{itcs/AlonGPRWW025}.

\paragraph{Complementary improved upper bounds.}
We further complement the lower bounds with improved non-adaptive algorithms for laminar and graphic matroids. The non-adaptive OCRS-style framework behind the tight result for truncated $R_1$-partition matroids can be extended to more general laminar-type constraints, and gives a $3$-competitive algorithm for truncated partition matroids and a $6.311$-competitive algorithm for laminar matroids. 

For graphic matroids, we use a different idea: instead of the standard ex-ante relaxation over the matroid polytope, we consider an alternative relaxation tailored to spanning forests, which leads to an $8$-competitive algorithm for graphic matroids, and the ratio improves to $3+2\sqrt{2}\approx5.828$ for simple graphs with no parallel edges.

\subsection{Related Work}

\paragraph{Prophet inequalities and Bayesian mechanism design.} Prophet inequalities have become a central tool in online stochastic optimization and Bayesian mechanism design. The connection to mechanism design was initiated by~\citet{aaai/HajiaghayiKS07} and developed systematically through sequential posted pricing~\cite{stoc/ChawlaHMS10, focs/DuettingFKL17}. We do not attempt to provide an exhaustive survey of this broad literature, and refer the reader to the surveys of~\citet{sigecom/Lucier17} and~\citet{sigecom/CorreaFHOV18}.

\paragraph{Other online selection models.} Beyond the classic matroid prophet inequality setting, there are three other commonly studied online selection models, namely matroid secretary~\cite{soda/BabaioffIK07, focs/Lachish14, FeldmanSZ15}, prophet matroid secretary~\cite{soda/EhsaniHKS18}, and sample-based matroid prophet inequalities~\cite{geb/AzarKW19, soda/CaramanisDFFLLP22, corr/abs-2103-13089, sigecom/0001LTWW024, soda/FeldmanSZ26, sosa/NutiW26}. While it remains unclear whether matroid secretary admits an $O(1)$-competitive algorithm, for graphic and laminar matroids this is true and there is a long line of work improving the competitive ratio~\cite{icalp/KorulaP09, soda/ImW11, stacs/Ma0W13, JailletSZ13, soda/SotoTV18, esa/BanihashemHKKMO25, ipco/BercziLSV25, esa/0002PZ24}. In particular, as noted by~\citet{dam/PashkovichS26}, for laminar and truncated partition matroids, the single-sample prophet inequalities in~\citet{geb/AzarKW19, corr/abs-2103-13089}, which rely on a restricted class of secretary algorithms, imply non-adaptive prophet inequalities since they set non-adaptive thresholds based on samples. Our improved algorithms are designed more directly and make better use of the full distributional information.

\paragraph{Online contention resolution schemes (OCRS).} OCRS was originally proposed by~\citet{soda/FeldmanSZ16,siamcomp/FeldmanSZ21} as a general framework to convert ex-ante feasible activation probabilities into online selections, and has wide applications in prophet inequalities, stochastic probing, and Bayesian mechanism design. Since then, OCRS has been widely studied in various settings: matroids~\cite{esa/LeeS18, itcs/DinevW24}, matching~\cite{sigecom/EzraFGT20, icalp/FuTWWZ21, sigecom/PollnerRSW22, soda/MacRuryMG23, sigecom/MaMN24, stoc/MacRuryM24, itcs/MaMZ26, corr/abs-2603-21532}, combinatorial auctions~\cite{wine/ChawlaRTT23, soda/0001CDHKR25, itcs/MaMZ26}, etc. There is also growing interest in OCRS under other online arrival models, such as random-order~\cite{focs/Adamczyk018, soda/MacRuryMG23, stoc/MacRuryM24, siamcomp/Dughmi25} and preselected-order~\cite{icalp/Zhao25}. 

While standard OCRSs are typically required to work for arbitrary unknown order or uniformly random arrival order, and are usually adaptive, our algorithm for truncated partition matroids can be viewed as a variant of OCRS where the arrival order is known, but the algorithm must use non-adaptive thresholds.

\section{Preliminaries}
\label{sec:prelim}

\subsection{Matroid Prophet Inequalities}

\begin{definition}[Matroid]
A matroid $M=(N,\mathcal I)$ consists of a ground set $N$ and a family of independent sets $\mathcal I\subseteq 2^N$ satisfying:
\begin{enumerate}
\item if $I\subseteq J$ and $J\in\mathcal I$, then $I\in\mathcal I$;
\item if $I,J\in\mathcal I$ and $|J|>|I|$, then there exists $j\in J\setminus I$ such that $I\cup\{j\}\in\mathcal I$.
\end{enumerate}
\end{definition}

A matroid prophet inequality instance is given by a matroid $M=([n],\mathcal I)$ and independent nonnegative random variables $X_1,\ldots,X_n$, where $X_i$ denotes the value of item $i$. The items arrive in order $1,2,\ldots,n$. When item $i$ arrives, the algorithm observes the realization of $X_i$ and must immediately and irrevocably decide whether to accept it, subject to the accepted set remaining independent in $M$.

\begin{definition}[Competitive ratio]
    Given an algorithm $\alg$, let $\E[\alg]$ denote the expected total value of accepted items, and $\E[\opt] = \E\left[\max_{I\in\mathcal{I}}\sum_{i\in I}X_i\right]$ denote the prophet's value, i.e. the expected value of the maximum independent set in hindsight. Then we say that $\alg$ has competitive ratio $\rho$ if for every instance, $\E[\opt]\leq \rho\cdot \E[\alg]$.
\end{definition}

\subsection{Matroid Classes}
\label{sec:matroid_classes}

\begin{definition}[Laminar matroid]
A family $\mathcal L\subseteq 2^N$ is laminar if for every $A,B\in\mathcal L$, either $A\cap B=\emptyset$, $A\subseteq B$, or $B\subseteq A$. A laminar matroid is specified by a laminar family $\mathcal L$ and integer ranks $(r_S)_{S\in\mathcal L}$; a set $I\subseteq N$ is independent if $|I\cap S|\leq r_S$ for every $S\in\mathcal L$. For convenience, we assume that no constraint is redundant: for every $A,B\in\mathcal{L}$, if $A\subseteq B$ then $r_A\leq r_B$.
\end{definition}

\begin{definition}[Truncated partition matroid]
A truncated partition matroid is the special case of a laminar matroid where the laminar family consists of a partition $\mathcal P$ of $N$ and the whole ground set $N$. Given local ranks $(r_B)_{B\in\mathcal P}$ and a global rank $k$, a set $I\subseteq N$ is independent if $|I|\leq k$ and $|I\cap B|\leq r_B$ for every $B\in\mathcal P$. If every local rank is $1$, we call it a truncated R1-partition matroid. 
\end{definition}

\begin{definition}[Graphic matroid]
A graphic matroid is defined over the edge set of a graph $G=(V,E)$, and a set of edges is independent if it forms a forest. Unless explicitly stated, graphs are allowed to have parallel edges. We use simple graphic matroids to refer to graphic matroids defined on a simple graph.
\end{definition}

\subsection{Non-Adaptive Algorithms}

The celebrated $2$-competitive algorithm of~\citet{stoc/KleinbergW12} is adaptive: the decision for each item may depend on previous realizations and decisions.

In this paper, we focus on non-adaptive algorithms, which choose thresholds $T_1,\ldots,T_n$ before any value is realized. These thresholds may only depend on the distributions, the matroid, and the arrival order. When item $i$ arrives, the algorithm accepts it if and only if $X_i\geq T_i$ and adding $i$ preserves independence of the set of accepted items. Note that the thresholds for the items can be different. A more restricted setting where all thresholds must be identical, namely static threshold algorithms, is discussed in Appendix~\ref{sec:discussion}.

\subsection{Continuous Distributions, Tie-Breaking, and Randomization}

For notational convenience, our algorithms are described for continuous distributions. Thus, for every item $i$ and every $q\in[0,1]$, there exists a threshold $\tau_i(q)$ such that $\Pr[X_i\geq \tau_i(q)]=q$. For discrete distributions with atoms, we can apply a standard randomized tie-breaking at the threshold so that item $i$ is considered above the threshold with probability exactly $q$.

In lower-bound proofs, it suffices to consider discrete distributions and deterministic thresholds. Indeed, any ex-ante random choice of thresholds (including randomness of tie-breaking) can be viewed as a distribution over deterministic threshold vectors. Thus, if every deterministic non-adaptive algorithm has competitive ratio at least $\rho$ on an instance, then so does every randomized non-adaptive algorithm.

\section{Truncated Partition and Laminar Matroids}
\label{sec:laminar}

\subsection{The Simplest Separation: Truncated \texorpdfstring{$R_1$}{R1}-Partition Matroids}
\label{sec:truncated_partition_lower_bound}

\begin{theorem}
\label{thm:truncated_lower_bound}
There exists a truncated $R_1$-partition matroid prophet inequality instance such that every non-adaptive algorithm has competitive ratio at least $2.179$.
\end{theorem}

\begin{proof}
    We construct an instance that implements the basic intuition: use a sequence of Bernoulli variables to create additional risk to the classic two-variable hard instance. 
    
    \paragraph{Construction.} Construct an instance with $N+2$ items $[N]\cup \{g_1,g_2\}$. The first $N$ items have identical distributions and take value $X_i=r \in(0,1)$ with probability $b$, and $X_i=0$ otherwise.  The remaining two ``important items'' $g_1,g_2$ form a classical two-variable hard instance, where $X_{g_1}=1$, and $X_{g_2}=1/\epsilon$ with probability $\epsilon$, and $0$ otherwise. We choose the parameters such that $b\to0, Nb\to\infty, \epsilon\to0$, and $r$ will be optimized later.
    
    The arrival order is $1,2,\ldots,N,g_1,g_2$. Each item $i\in[N]$ forms a singleton partition, and $\{g_1,g_2\}$ forms the last partition. Each partition has local rank $1$ and the global rank is $2$.

    \paragraph{Prophet value.} Subject to the constraints, the prophet takes the better of $g_1,g_2$, which gives expected value $2-\epsilon$.  It also obtains value $r$ from the first $N$ items whenever at least one of them realizes value $r$.  Therefore
    \[
    \E[\opt] = 2-\epsilon+r(1-(1-b)^N) ~,
    \]
    which approaches $2+r$ as $b\to0$, $Nb\to\infty$, and $\epsilon\to0$.
    
    \paragraph{Non-adaptive algorithms.}  Essentially, a deterministic non-adaptive algorithm only needs to decide for how many items in $[N]$ it should set a positive threshold. Let $\alg_M$ denote the algorithm that sets $M$ positive thresholds, and let $Z\sim\operatorname{Bin}(M,b)$ be the number of items among the $M$ chosen items that realize value $r$. The algorithm gains value $r\min\{Z,2\}$ from $[N]$, and only when $Z\leq 1$ the algorithm can in addition gain expected value $1$ from $\{g_1,g_2\}$. Therefore, for each $M$, we have
    \[
    \E[\alg_M] = r\E[\min\{Z,2\}] + \Pr[Z\le1] = 2r+(1-2r)(1-b)^M+(1-r)Mb(1-b)^{M-1} ~.
    \]
    
    Recall that we take $b\to0$. Suppose $Mb\to\lambda$ for some constant $\lambda\ge0$, then $(1-b)^M\to e^{-\lambda}$, and $Mb(1-b)^{M-1}\to \lambda e^{-\lambda}$, hence
    \[
    \E[\alg_M]\to 2r+e^{-\lambda}\bigl((1-2r)+(1-r)\lambda\bigr) := V_r(\lambda) ~.
    \]
    
    The optimal algorithm chooses $M$ such that $V_r(\lambda)$ is maximized. Differentiating $V_r(\lambda)$ gives
    \[
    \frac{d}{d\lambda}\left[2r + e^{-\lambda}\bigl((1-2r)+(1-r)\lambda\bigr)\right] = e^{-\lambda}(r-(1-r)\lambda) ~,
    \]
    thus the maximum is attained at $\lambda^*=\frac{r}{1-r}$, and $V_r(\lambda^*) = 2r+(1-r)\exp\left(-\frac{r}{1-r}\right)$.
    
    Therefore, for each $r$, the instance implies an asymptotic lower-bound ratio of
    \[
    R(r)=\frac{2+r}{2r+(1-r)\exp(-r/(1-r))} ~.
    \]
    It remains to optimize $r$. Differentiating $R(r)$ gives the first-order optimality condition
    \[
    (5-2r)\exp\left(-\frac{r}{1-r}\right)=4(1-r) ~.
    \]
    This equation has a unique solution $r^*\approx 0.3211$ in $(0,1)$, and gives $R(r^*) \approx 2.179$.
\end{proof}

\begin{remark}
    This construction can be converted to a transversal matroid instance that implies the same lower bound of $2.179$: consider a bipartite graph where the left side consists of $N+2$ vertices $[N]\cup \{g_1,g_2\}$, and the right side consists of two vertices $\{h_1,h_2\}$. Every vertex $i\in[N]$ has an edge to both $h_1$ and $h_2$, while $g_1,g_2$ each has one edge to $h_2$. It can be verified that this is equivalent to the original truncated partition matroid. 
\end{remark}

\subsection{Tight Upper Bound for Truncated \texorpdfstring{$R_1$}{R1}-Partition Matroids}
\label{sec:truncated_R1_partition_upper_bound}

We settle the cost of non-adaptivity in the special case of truncated $R_1$-partition matroids. We relax the prophet benchmark via the standard ex-ante relaxation, then analyze a non-adaptive OCRS-style process. The resulting upper bound exactly matches the lower bound in Theorem~\ref{thm:truncated_lower_bound}. 

\paragraph{Ex-ante relaxation.}
For each item $i$, define its expected top-quantile value $\nu_i:[0,1]\to\mathbb R_{\ge0}$ by $\nu_i(q):=\E[X_i\cdot \mathbf 1\{X_i\ge \tau_i(q)\}]$.
The function $\nu_i$ is non-decreasing and concave.

The matroid polytope of a matroid $\mathcal{M}=(N,\mathcal{I})$ is the convex hull of the indicator vectors of the independent sets of $\mathcal{M}$, formally $\mathcal{P}(\mathcal{M})=\mathrm{conv}(\{e_I\mid I\in \mathcal{I}\})$.  The following ex-ante relaxation to the matroid polytope upper bounds the expected value of the prophet:
\begin{equation}
\label{eq:exante_relaxation}
    \E[\opt]\le \max_{x\in P(\mathcal M)}\sum_i \nu_i(x_i) ~.
\end{equation}
Indeed, suppose the prophet selects item $i$ with probability $x'_i$, then $x'\in P(\mathcal{M})$ and the expected value from $i$ is at most $\nu_i(x'_i)$.

\paragraph{Non-adaptive OCRS.} We use the following algorithmic framework. Fix an ex-ante solution $x\in P(\mathcal M)$ that maximizes~\eqref{eq:exante_relaxation} and a parameter $\alpha \in (0,1)$. Choose the threshold for each item in arrival order. For item $i$, let $F_i$ be the probability that item $i$ is feasible when it arrives, under the thresholds already fixed for previous items. Set the threshold of item $i$ to $\tau_i\left(\frac{\alpha x_i}{F_i}\right)$. \footnote{Note that the algorithm is similar to the one used in~\citet{sigecom/EzraFGT20} for matching with edge arrival. In their paper, they say that their algorithm is adaptive, since they compute $F_i$ when $i$ arrives, and thus the analysis also holds if the arrival order is not known in advance. In this paper, we always assume the arrival order is known, so $F_i$ can be computed offline in advance.}

In other words, conditioned on item $i$ remaining feasible (i.e., after including $i$ the accepted subset remains independent) when it arrives, the algorithm accepts it with probability $\frac{\alpha x_i}{F_i}$. Thus, item $i$ is accepted with marginal probability $\alpha x_i$. Moreover, by concavity of $\nu_i$, the expected value from $i$ is $F_i\nu_i\left(\frac{\alpha x_i}{F_i}\right) \geq F_i \cdot \frac{\alpha}{F_i}\nu_i(x_i) =\alpha\nu_i(x_i)$. Therefore, if for any $x\in P(\mathcal{M})$, we have $F_i\ge\alpha$ for every item $i$, then the algorithm is well-defined and is $1/\alpha$-competitive, since
\[
    \E[\alg]\ge \alpha\sum_i \nu_i(x_i)\ge \alpha\E[\opt] ~.
\]

\begin{theorem}
\label{thm:truncated-R1-partition}
There is a non-adaptive $\approx 2.179$-competitive algorithm for truncated $R_1$-partition matroid prophet inequalities.
\end{theorem}

The proof relies on the following technical tools. After relating the event that an item becomes infeasible to a Poisson binomial variable, we apply Lemma~\ref{lem:capped_poisson_binomial} to upper bound its probability. Then Lemma~\ref{lem:alpha_tight_all_k} justifies why the tight ratio exactly matches the lower bound ratio from Theorem~\ref{thm:truncated_lower_bound}.

\begin{definition}
\label{def:zeta}
For $\lambda\geq 0$, let $P_\lambda\sim\pois(\lambda)$, and define
\[
    \mu_k(\lambda):=\E[\min\{P_\lambda,k\}].
\]
For $m\in[0,k)$, let $\lambda_k(m)$ be the unique value satisfying $\mu_k(\lambda_k(m))=m$, and set $\lambda_k(k)=\infty$. Define
\[
    \zeta_k(m):=\Pr[P_{\lambda_k(m)}\geq k].
\]
Equivalently, if $\mu_k(\lambda)=m$, then $\zeta_k(m)=\Pr[P_\lambda\geq k]$.
\end{definition}

\begin{lemma}
\label{lem:capped_poisson_binomial}
    Let $X$ be a Poisson binomial random variable, and suppose a random variable $Z$ and an event $\mathcal E$ satisfy $\mathcal E \subseteq \{X\geq k\}$ and $Z=\min\{X,k-1\}+\mathbf 1_{\mathcal E}$, then $\Pr[\mathcal E]\leq \zeta_k(\E[Z])$.
\end{lemma}

\begin{lemma}
\label{lem:alpha_tight_all_k}
Let $R(r)=\frac{2+r}{2r+(1-r)\exp(-r/(1-r))}$ be the lower bound ratio from Theorem~\ref{thm:truncated_lower_bound}, and $r^*$ be the unique maximizer of $R(r)$ on $(0,1)$. Set $\alpha^*:=1/R(r^*)$, then for every $k\geq 2$ and every $a\in[0,1]$, we have
\[
    G_k(a) = 1-\alpha^*a-\zeta_k(\alpha^*(k-a))\geq \alpha^*.
\]
\end{lemma}
These two lemmas will be restated and proved in Appendix~\ref{sec:appendix_truncated_deferred_proofs}.

\begin{proofof}{Theorem~\ref{thm:truncated-R1-partition}}
Let $k$ denote the global rank. If $k=1$ then the instance is equivalent to single-choice prophet inequality, and there is nothing to prove. Now assume $k\geq 2$. Fix any $x\in P(\mathcal{M})$ and run the non-adaptive OCRS procedure with $\alpha=\alpha^*=1/R(r^*)$. It suffices to prove inductively that $F_i\ge\alpha^*$ for every item $i$. 

Fix $i$ and assume $F_j\geq \alpha^*$ for all $j<i$, so each previous item $j$ is accepted with probability $\alpha^* x_j$. Suppose item $i$ belongs to partition $B$. We refer to other partitions and items in those partitions as ``external''. Define
\[
    a:=\sum_{j<i,\ j\in B}x_j, \qquad o:=\sum_{j<i,\ j\notin B}x_j.
\]
Since $x\in P(\mathcal M)$, we have $a\leq 1$ and $a+o\leq k$.

Item $i$ can become infeasible if one of the following two disjoint events occurs:
\begin{enumerate}[leftmargin=*]
    \item Some item $j\in B$ that arrives before $i$ is accepted. By inductive hypothesis, for each $j\in B,j<i$ this happens with probability $\alpha^* x_j$, and this cannot simultaneously happen for two items in $B$. Thus the total probability of this event is exactly $\alpha^* a$.
    \item Exactly $k$ external items that arrive before $i$ are accepted. Let $Z$ be the number of accepted external items before item $i$ arrives, then $\E[Z]=\alpha^* o$. We can obtain a strong upper bound on the probability that $Z=k$, since $Z$ is not concentrated around $Z=0$ and $Z=k$, but is related to a Poisson binomial variable in the following way. 
    
    For each external partition $C\neq B$, let $Y_C$ be the indicator variable that at least one previous item in $C$ realizes above its thresholds (regardless of whether the global rank is full), and $D:=\sum_{C\ne B}Y_C$, then $D$ is a Poisson binomial variable. If $D\leq k-1$ then at most $k-1$ external elements can be accepted, so we have $Z=\min\{D,k-1\} + \mathbf 1_{Z=k}$. Combining this relationship and $\E[Z]=\alpha^* o$ we get $\Pr[Z=k]\le \zeta_k(\alpha^* o)$ by Lemma~\ref{lem:capped_poisson_binomial}.
\end{enumerate}

Therefore, by union bound, we can lower bound the probability that $i$ remains feasible:
\[
    F_i \geq 1-\alpha^* a-\zeta_k(\alpha^* o) \geq 1-\alpha^* a-\zeta_k(\alpha^*(k-a)) \geq \alpha^* ~.
\]
The second inequality follows from monotonicity of $\zeta_k$, and the last inequality is by Lemma~\ref{lem:alpha_tight_all_k}.
\end{proofof}

\subsection{Improved Bounds for Laminar Matroids}

For laminar matroids, we can further recursively extend the basic construction of Theorem~\ref{thm:truncated_lower_bound} into a general laminar matroid and obtain an improved lower bound. The details are deferred to Appendix~\ref{sec:appendix_laminar_lower_bound}.

\begin{theorem}
\label{thm:laminar_lower_bound}
There exists a laminar matroid prophet inequality instance such that every non-adaptive algorithm has competitive ratio at least $2.217$.
\end{theorem}

Our tight upper bound analysis for truncated $R_1$-partition matroids cannot be directly extended to more complex capacity constraints. For general truncated partition matroids, where local ranks can be larger than $1$, we have a coarser analysis that simply upper bounds the probability of the two events that make an item infeasible using Markov inequalities. For laminar matroids, there are more such events, and we can only analyze a weaker algorithm that sets the thresholds for each item independently. The proofs are deferred to Appendix~\ref{sec:appendix_laminar_upper_bound}, where we also discuss the algorithmic aspects of computing $F_i$ efficiently.

\begin{theorem}
\label{thm:truncated_upper_bound}
There is a $3$-competitive non-adaptive algorithm for truncated partition matroid prophet inequalities.
\end{theorem}

\begin{theorem}
\label{thm:laminar_upper_bound}
There is a $6.311$-competitive non-adaptive algorithm for laminar matroid prophet inequalities.
\end{theorem}

\section{Graphic Matroids}
\label{sec:graphic}

\subsection{Lower Bound for Graphic Matroids}

In this section we demonstrate for the first time a strict cost of non-adaptivity in graphic matroid prophet inequalities: there exist instances where any non-adaptive algorithm has competitive ratio strictly greater than $2$, whereas adaptive algorithms are $2$-competitive. We first give a basic construction that already gives a lower bound of $\frac{5}{2}$, and later show how it can be modified to obtain a stronger lower bound of $3$. Note that our construction crucially relies on using parallel edges and does not apply to simple graphs.

\begin{theorem}
\label{thm:graphic_lower_bound_5/2}
    There exists a family of graphic matroid prophet inequality instances such that any non-adaptive algorithm has competitive ratio at least $\frac{5}{2}$.
\end{theorem}
\begin{proof}
    We use a structure inspired by the ``hat graph'', which has been used to demonstrate the challenge of designing OCRS and matroid secretary algorithms for graphic matroids~\cite{esa/LeeS18, wine/BahraniBSW21}.

    \paragraph{Construction.} Let $n$ be a sufficiently large integer. Consider the following graph $G=(V,E)$. The vertices are $V=\{s,t\}\cup \{v_1,v_2,\ldots,v_n\}$. Let $p = n^{-1/3}$, and the edges are:
    \begin{itemize}
        \item For each $i\in[n]$, there are two edges $e_{i,1}$ and $e_{i,2}$ between $s,v_i$. $X_{e_{i,1}}$ is a deterministic $1$, while $X_{e_{i,2}}$ takes value $1/p$ with probability $p$ and $0$ otherwise. There is one edge $e_{i,3}$ between $t,v_i$, and $X_{e_{i,3}}$ takes value $1/p$ with probability $p$ and $0$ otherwise. These three edges are referred to as ``hat $i$'' or ``hat edges''.
        \item There are two edges $e_{b,1}$ and $e_{b,2}$ between $s,t$. $X_{e_{b,1}}$ has deterministic value $n$, and $X_{e_{b,2}}$ takes value $n/\epsilon$ with probability $\epsilon$ and $0$ otherwise, for some sufficiently small $\epsilon>0$ (think of $\epsilon = 1/poly(n)$). These two edges are referred to as ``brim edges''.
    \end{itemize}
    Figure~\ref{fig:single-layer-hat} provides an illustration of the construction. The edges arrive in the following order: for each $i$ sequentially, $e_{i,1},e_{i,2},e_{i,3}$ arrive one after another; finally $e_{b,1}$ and $e_{b,2}$ arrive one after another.

    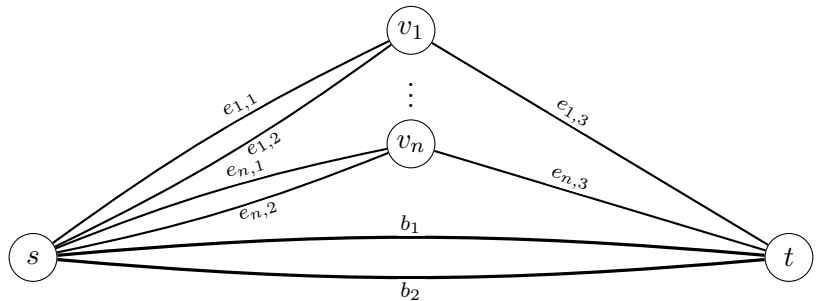
\begin{figure}[ht]
    \centering
    \begin{tikzpicture}[
        x=1cm, y=1cm,
        vertex/.style={circle, draw, fill=white, inner sep=1.2pt, minimum size=18pt},
        edge/.style={draw, thick},
        brim/.style={draw, very thick},
        elabel/.style={font=\scriptsize, fill=white, inner sep=1pt}
    ]
    
    \node[vertex] (s) at (-5,0) {$s$};
    \node[vertex] (t) at (5,0) {$t$};
    
    \draw[brim, bend left=5] (s) to node[elabel, above] {$b_1$} (t);
    \draw[brim, bend right=5] (s) to node[elabel, below] {$b_2$} (t);
    
    \node[vertex] (v1) at (0,3) {$v_1$};
    \node[vertex] (vn) at (0,1.5) {$v_n$};
    
    \node at (0,2.25) {$\vdots$};
    
    \draw[edge] (s) to[bend left=5]
        node[elabel, pos=0.6, above, sloped] {$e_{1,1}$} (v1);
    \draw[edge] (s) to[bend right=5]
        node[elabel, pos=0.6, below, sloped] {$e_{1,2}$} (v1);
    \draw[edge] (t) --
        node[elabel, pos=0.6, above, sloped] {$e_{1,3}$} (v1);
    
    
    \draw[edge] (s) to[bend left=5]
        node[elabel, pos=0.6, above, sloped] {$e_{n,1}$} (vn);
    \draw[edge] (s) to[bend right=5]
        node[elabel, pos=0.6, below, sloped] {$e_{n,2}$} (vn);
    \draw[edge] (t) --
        node[elabel, pos=0.6, above, sloped] {$e_{n,3}$} (vn);
    
    \end{tikzpicture}
    \caption{The hat example.}
    \label{fig:single-layer-hat}
    \end{figure}

    \paragraph{Prophet Value.} There are essentially two types of spanning trees of the graph: either take one brim edge and one edge from each hat, or take two edges from one hat and one edge from every other hat. The expected maximum value of the brim edges is 
    \[
    \E[\max\{X_{e_{b,1}}, X_{e_{b,2}}\}] = n\cdot (1-\epsilon) + n/\epsilon\cdot \epsilon = (2-\epsilon)n ~,
    \]
    and taking $e_{b,1}$ guarantees value at least $n$, while giving up the brim can only get one more hat edge in return, which gives value at most $1/p=n^{1/3}$. Thus, the maximum spanning tree must be of the first type and includes exactly one brim edge. 
    
    For each hat, the expected maximum value among the three hat edges is 
    \[
    \E[\max\{X_{e_{i,1}}, X_{e_{i,2}}, X_{e_{i,3}}\}] = 1\cdot (1-p)^2 + 1/p\cdot (1-(1-p)^2) = 3-3p+p^2 ~.
    \]
    Therefore, summing over the brim edges and all $n$ hats, we have
    \[
    \E[\opt] = (5-3p+p^2-\epsilon)n ~.
    \]
    
    \paragraph{Non-Adaptive Algorithms.} There are four different non-adaptive strategies for each hat $i$: 
    \begin{itemize}
        \item Strategy 1: ignore all three edges. Expected value $0$.
        \item Strategy 2: always accept one edge and ignore the other two edges. Expected value $1$.
        \item Strategy 3: accept $e_{i,1}$, then ignore $e_{i,2}$ and set threshold $1/p$ for $e_{i,3}$. Expected value $2$.
        \item Strategy 4: ignore $e_{i,1}$, and set threshold $1/p$ for both $e_{i,2}$ and $e_{i,3}$. Expected value $2$. 
    \end{itemize}
    Strategy 1 is dominated by Strategy 2 since taking one edge from a hat does not affect other hats or the brim edges. Strategy 3 is dominated by Strategy 4 because they have the same expected value while Strategy 4 has lower probability of making $s,t$ connected.

    Therefore, the optimal non-adaptive algorithm essentially depends on how many hats use Strategy 2 and how many hats use Strategy 4. This is the major dilemma that the non-adaptive algorithm must face: if it does not use Strategy 4 very often, then it obtains low expected reward from the hats; if it uses Strategy 4 very often, then there is a high probability that $s,t$ become connected before the brim edges arrive.
    
    Let $\alg_k$ denote an algorithm where $k$ hats use Strategy 4, and w.l.o.g. assume they are the first $k$ hats. For each hat that uses Strategy 4, conditioned on $s,t$ not being connected before the hat arrives, $s,t$ become connected after the hat arrives with probability $p^2$. Let $q = 1-p^2$, then after the first $j$ hats arrive, the probability that $s,t$ are still not connected is $q^{j}$.

    For each of the first $k$ hats, if $s,t$ are not connected when it arrives, Strategy 4 gets expected value $2$; otherwise only one edge can be accepted, and the expected value is $2-p$. Therefore, the expected value from the hats is
    \begin{align*}
    (n-k) + \sum_{j=1}^k\left(q^{j-1}\cdot 2 + (1-q^{j-1})\cdot (2-p)\right)
    = &~ (n-k) + (2-p)k + p\sum_{j=0}^{k-1}q^j\\
    = &~ n+k-pk + p\frac{1-q^k}{1-q} = n+k-pk + \frac{1-q^k}{p} ~.
    \end{align*}
    After all hats have arrived, $s,t$ are still not connected with probability $q^k$, and the algorithm can get additional expected value $n$ by accepting one brim edge. We conclude that
    \[
    \E[\alg_k] = n+\frac{1}{p} + (1-p)k + \left(n-\frac{1}{p}\right)q^k ~.
    \]

    Notice that this is convex in $k$, so the maximum is achieved at either $k=0$ or $k=n$. We have $\E[\alg_0] = 2n$ and
    \begin{align*}
    \E[\alg_n] = &~ n+n^{1/3} + (1-n^{-1/3})n + (n-n^{1/3})\left(1-n^{-2/3}\right)^n\\
    \leq &~ 2n -(n^{2/3}-n^{1/3}) + (n-n^{1/3})\cdot e^{-n^{1/3}} < 2n ~.
    \end{align*}
    The last inequality is because $n^{-1/3} \gg e^{-n^{1/3}} \implies n^{2/3}-n^{1/3} \gg (n-n^{1/3})\cdot e^{-n^{1/3}}$.

    Thus, the optimal non-adaptive algorithm has expected value at most $2n$, hence the competitive ratio when $n\to\infty$ is
    \[
    \lim_{n\to\infty}\frac{\E[\opt]}{\E[\alg_0]} = \lim_{n\to\infty}\frac{(5-3n^{-1/3}+n^{-2/3}-\epsilon)n}{2n} = \frac{5}{2}-\frac{\epsilon}{2} ~. \qedhere
    \]
\end{proof}

The lower bound can be further strengthened: for each pair of hat edges $e_{i,1},e_{i,2}$, we can view them as brim edges and recursively apply the construction. The only modification is that the variables on the internal brim edges must be selected carefully to enable a recursive analysis. The detailed recursive construction is deferred to Appendix~\ref{sec:appendix_graphic_recursive}.

\begin{theorem}
\label{thm:graphic_lower_bound_3}
    For any integer $k\geq 1$, there exists a family of graphic matroid prophet inequality instances such that any non-adaptive algorithm has competitive ratio at least $3-\frac{1}{k+1}$.
\end{theorem}

\subsection{A Simple Algorithm}

We use the same random partitioning idea from prior work~\cite{sagt/ChawlaGKM24, dam/PashkovichS26}, but instead of relying on the ex-ante relaxation, we compare the algorithm against a different relaxation of $\opt$, which turns out to be simpler to analyze and also facilitates an improved analysis for simple graphs. 
Given a forest of a graph, we can arbitrarily root each tree, then charge each edge to the endpoint with larger depth. In this way, each vertex is charged at most one edge. Fix a realization $\vec X$, suppose in the optimal spanning forest, edge $e(\vec X,v)$ is charged to $v$. Then we can upper bound $\E[\opt]$ by
\begin{align*}
\E[\opt] \leq\E\left[\sum_{v\in V}X_{e(\vec X, v)}\right]
\leq  \E\left[\sum_{v\in V}\max_{e\ni v}X_e\right] = \sum_{v\in V}\E\left[\max_{e\ni v}X_e\right] ~.
\end{align*}
That is, instead of accepting each edge with probability proportional to the ex-ante relaxation, it suffices to accept an incident edge for each vertex that is comparable to the maximum incident edge in expectation, so this effectively reduces the problem to standard single-choice prophet inequalities.

Our algorithm randomly partitions the graph into two parts $A,B$, then for each $v\in B$, run an independent prophet inequality algorithm on only the edges incident to $v$ such that the other endpoint is in $A$. Since we never accept edges internal to $A$ or $B$, for each $v\in B$ we can always accept at least one incident edge without violating the graphic matroid constraint, which enables us to obtain a $1/2$-approximation to the prophet inequality subproblem.

\begin{algorithm}
    \caption{Improved non-adaptive algorithm for graphic matroid}
    \label{alg:alg1}
    \KwIn{Graph $G = (V,E)$, variables $(X_e)_{e\in E}$.}
    \KwOut{A vector of thresholds $(T_e)_{e\in E}$.}
    $A\gets \text{a random subset of $V$ that includes each vertex with probability $1/2$}$\;
    $B\gets V\setminus A$\;
    For each $v\in B$, compute $T_v$ such that $\Pr[\max_{e=(v,u),u\in A}X_e\geq T_v]=1/2$. Then for each $e=(v,u),u\in A$, set $T_e = T_v$\;
    For all other edges with both endpoints in $A$ or both endpoints in $B$, set $T_e=+\infty$.
\end{algorithm}

The analysis relies on the following fact.

\begin{fact}
\label{fact:random_subset}
    If a randomized subset $S\subseteq [n]$ includes each $i\in[n]$ with marginal probability $1/2$, then for any sequence of random variables $X_1,\ldots,X_n$, we have
    \[
    \E_S\left[\E\left[\max_{i\in S}X_i\right]\right] \geq \frac{1}{2}\cdot \E\left[\max_{i\in[n]}X_i\right] ~.
    \]
\end{fact}

Therefore, informally speaking, we lose three factors of $2$ due to $\Pr[v\in B]=1/2$, Fact~\ref{fact:random_subset}, and the prophet inequality approximation respectively. 

\begin{theorem}
\label{thm:graphic_upper_bound_8}
    Algorithm~\ref{alg:alg1} is  $8$-competitive for graphic matroid prophet inequalities.
\end{theorem}
The formal proof is deferred to Appendix~\ref{sec:appendix_graphic_upper}.

\begin{remark}
    The randomized partition can be derandomized using standard techniques, as stated in \citet{sagt/ChawlaGKM24}.
\end{remark}

\subsection{Improved Algorithm for Simple Graphs}

For simple graphs, the above analysis of Algorithm~\ref{alg:alg1} is not tight: we lose a factor of $2$ due to Fact~\ref{fact:random_subset} and another factor of $2$ due to prophet inequality, but it turns out that these two factors cannot be simultaneously tight when there are no parallel edges, i.e., the events that the edges incident to $v$ are active are independent. In fact, we can show a more general result: there is an algorithm that, given a full prophet inequality instance but each variable is presented independently with probability $p$ in the online process, obtains expected reward at least a $\frac{p}{p+1}$ fraction of the full prophet, instead of $\frac{p}{2}$. By optimizing over the partition probability, this eventually leads to an improved analysis for simple graphs. The details are deferred to Appendix~\ref{sec:appendix_graphic_upper}.

\begin{theorem}
\label{thm:graphic_upper_bound_simple}
    There is a $3+2\sqrt 2\approx 5.828$-competitive non-adaptive algorithm for simple graphic matroid prophet inequalities.
\end{theorem}

\section{Conclusion and Open Problems}

In this paper, we study the cost of non-adaptivity in matroid prophet inequalities, and demonstrate that for many natural classes of matroids, non-adaptivity introduces additional structural barriers and leads to competitive ratios strictly larger than $2$. We obtain a tight competitive ratio for the special class of truncated $R_1$-partition matroids, and reveal an interesting connection between non-adaptive algorithms and Poisson binomial variables.

For other matroids, our current results are not tight, and the immediate open problem is to close these gaps. In particular, we are most interested in the following two questions:
\begin{enumerate}
    \item Is there a lower bound for simple graphic matroids that is strictly larger than the adaptive benchmark of $2$? As shown in Appendix~\ref{sec:appendix_graphic_simple_fail}, our approach fails to extend to simple graphs, and new constructive ideas are needed. 
    \item We conjecture that for general truncated partition matroids the ratio of $\approx 2.179$ is also tight. Proving this would require a non-trivial extension of the Poisson binomial characterization for truncated $R_1$-partition matroids.
\end{enumerate}

\section*{Acknowledgements}
The author would like to thank Kamesh Munagala, Siddhartha Banerjee, and Yu Cheng for helpful discussions. 

The proofs in this paper are all human-written. ChatGPT-5.5 is used to assist paraphrasing and polishing the text, searching related work, plotting tables and figures, and generating the code for the numeric optimizations in the proof of Theorem~\ref{thm:laminar_lower_bound} and Theorem~\ref{thm:laminar_upper_bound}. In particular, it discovered the connection to prior lower bound on static threshold algorithms for matroid secretary, which led to Theorem~\ref{thm:static-partition-superconstant}.
All generated content was reviewed, verified, and finalized by the author who take full responsibility for the paper’s correctness.

\bibliographystyle{plainnat}
\bibliography{main}

\newpage
\appendix

\section{Discussion on \texorpdfstring{$k$}{k}-Uniform Matroids and Static Thresholds}
\label{sec:discussion}

In this section, we provide some observations that further complete the landscape of different types of algorithms for matroid prophet inequalities. For $k$-uniform matroids, non-adaptive algorithms are, in a worst-case sense, not more powerful than using a static threshold. Further, while static threshold algorithms have strong performance for $k$-uniform matroids, they turn out to be too restrictive for more general classes of matroids: already for partition matroids, static thresholds are not constant-competitive. The results in this section mainly follow from known results that have not been explicitly stated in the non-adaptive setting, and should not be viewed as our novel contributions.

\subsection{\texorpdfstring{$k$}{k}-uniform matroids}

For $k$-uniform matroids, we \textbf{temporarily} change to the notation that the competitive ratio is the ratio between the algorithm and the prophet, so it is always a number less than or equal to $1$, and a larger number is better. This is aligned with the conventional choice of prior papers on this line of work. 

It has long been known that the competitive ratio approaches $1$ as $k\to\infty$ \cite{aaai/HajiaghayiKS07, focs/Alaei11}. For static threshold algorithms, \citet{aaai/HajiaghayiKS07} gives an asymptotic competitive ratio of $1-O(\sqrt{\log k / k})$, while for small values of $k$, \citet{ior/ChawlaDL24} propose an algorithm such that for all $k\geq 2$, the competitive ratio is strictly better than $1/2$. This algorithm is later shown to be optimal among static threshold algorithms by \citet{ec/JiangMZ23}. For adaptive algorithms, \citet{focs/Alaei11} gave an algorithm with competitive ratio $1-1/\sqrt{k+3}$; \citet{ior/JiangMZ25} further improve upon this bound for all $k>1$, and demonstrate that for all $2\leq k\leq 50$, their adaptive algorithm strictly outperforms the optimal static threshold algorithm. For example, when $k=2$, the static threshold algorithm has competitive ratio $\approx 0.5859$, while the adaptive algorithm has competitive ratio $\approx 0.6148$. It is reasonable to believe that the separation also exists for all $k>50$. 

The power and limits of non-adaptive algorithms have not been explicitly considered before in the setting of $k$-uniform matroids. We observe that, perhaps surprisingly, allowing item-specific thresholds does not help if the thresholds must be fixed in advance: non-adaptive algorithms have at most the same worst-case competitive ratio as the optimal static threshold algorithm of \citet{ior/ChawlaDL24}. We provide an explicit counterexample, whose construction and analysis are structurally similar to the counterexample used in the proof of Theorem~\ref{thm:truncated_lower_bound}. Since static threshold algorithms are a subset of non-adaptive algorithms, this counterexample implies that the static threshold algorithm of \citet{ior/ChawlaDL24} is in fact optimal even among non-adaptive algorithms (in the worst-case sense). Previously, \citet{ec/JiangMZ23} proved its optimality among static threshold algorithms via an LP-based argument without explicitly constructing a counterexample, and our observation can be viewed as a concrete and simpler proof of a stronger tightness argument.

\begin{theorem}\label{thm:kuniform-nonadaptive}
For every $k\geq 2$, no non-adaptive algorithm for $k$-uniform matroid prophet inequalities can achieve a better competitive ratio than the optimal static threshold algorithm.
\end{theorem}

\begin{proof}
Let $P_\lambda\sim \pois(\lambda)$ with probability mass function $p_t(\lambda) = \Pr[P_{\lambda}=t] = e^{-\lambda}\frac{\lambda^t}{t!}$, and define
\[
\delta_k(\lambda):=\Pr[P_\lambda\leq k-1] ~,
\qquad
\mu_k(\lambda):=\frac{1}{k}\E[\min\{P_\lambda,k\}] ~.
\]
Let $\lambda_k$ be the unique solution of $\delta_k(\lambda_k)=\mu_k(\lambda_k)$, and let $\phi_k:=\delta_k(\lambda_k)=\mu_k(\lambda_k)$. According to \citet{ior/ChawlaDL24, ec/JiangMZ23}, the optimal static threshold algorithm has competitive ratio $\phi_k$. Next we show that for all $k\geq 2$, there is an instance on which every non-adaptive algorithm also obtains expected value at most $(\phi_k+o(1))\E[\opt]$.

\paragraph{Construction.}
Construct an instance with $N+1$ items $[N]\cup \{g\}$. 
\begin{itemize}
    \item The first $N$ items have identical distributions and take value $X_i=1$ with probability $b$, and $X_i=0$ otherwise.
    \item Let $h_k=\frac{\delta_k(\lambda_k)}{p_{k-1}(\lambda_k)}$. The final item $g$ takes value $h_k/\eps$ with probability $\eps$ and $0$ otherwise. 
\end{itemize} 
We choose the parameters such that $b \to 0,  Nb \to \infty, \eps \to 0$.

\paragraph{Prophet value.}
Since we choose $Nb\to\infty$, at least $k$ items in $[N]$ realize value $1$ with probability $1-o(1)$. Therefore, if $g$ takes value $0$, the prophet obtains expected value $k-o(1)$ from $[N]$. Otherwise, the prophet takes $g$ and $k-1$ largest items in $[N]$, obtaining expected value $h_k/\eps + k-1-o(1)$. Hence
\[
\E[\opt] = (1-\eps)(k-o(1))
+ \eps\left(\frac{h_k}{\eps} + k-1-o(1)\right) = k+h_k-o(1) ~.
\]

\paragraph{Non-Adaptive Algorithms.}
Essentially, a non-adaptive algorithm only needs to decide for how many items in $[N]$ it sets a positive threshold. Let $\alg_M$ denote the algorithm that sets $M$ positive thresholds, and let $Z\sim\operatorname{Bin}(M,b)$ be the number of items among the $M$ chosen items that realize value $1$. The algorithm gains value $\min\{Z,k\}$ from $[N]$, and only if $Z<k$ the algorithm can obtain additional expected value $h_k$ from $g$. Thus, for each $M$, we have
\[
\E[\alg_M]=\E[\min\{Z,k\}]+h_k\Pr[Z\le k-1] ~.
\]

Suppose $Mb\to\lambda$ for some constant $\lambda \geq 0$, recall that we take $b\to 0$, so $Z$ converges to $P_\lambda\sim\pois(\lambda)$, and hence
\[
\E[\alg_M]\to
\E[\min\{P_\lambda,k\}]+h_k\delta_k(\lambda):= V_k(\lambda) ~.
\]

It remains to maximize $V_k(\lambda)$. To differentiate $V_k(\lambda)$, we use the following elementary properties of Poisson variables:
\[
\frac{d}{d\lambda}\E[\min\{P_\lambda,k\}]=\delta_k(\lambda) ~,
\qquad
\frac{d}{d\lambda}\delta_k(\lambda)=-p_{k-1}(\lambda) ~.
\]
Therefore $V_k'(\lambda)=\delta_k(\lambda)-h_kp_{k-1}(\lambda)$. By definition of $h_k$, this is equal to $0$ when $\lambda = \lambda_k$.

Moreover, we have
\[
\frac{\delta_k(\lambda)}{p_{k-1}(\lambda)} = \sum_{t=0}^{k-1}\frac{p_t(\lambda)}{p_{k-1}(\lambda)} = \sum_{t=0}^{k-1}\frac{(k-1)!}{t!}\lambda^{t-(k-1)} ~.
\]
For $k\ge2$, this is strictly decreasing in $\lambda$, hence $\frac{\delta_k(\lambda)}{p_{k-1}(\lambda)}>
\frac{\delta_k(\lambda_k)}{p_{k-1}(\lambda_k)}
=h_k$ for all $\lambda < \lambda_k$, which implies $V'_k(\lambda)>0$; similarly we can show $V_k'(\lambda)<0$ for all $\lambda>\lambda_k$.

Thus we conclude that $V_k(\lambda)$ is uniquely maximized at $\lambda=\lambda_k$, and
\[
V_k(\lambda_k) = \E[\min\{P_{\lambda_k},k\}]+h_k\delta_k(\lambda_k) = (k+h_k)\cdot \phi_k ~.
\]

Therefore, for this instance, every non-adaptive algorithm satisfies
\[
\frac{\E[\alg]}{\E[\opt]} \leq \frac{(k+h_k)\cdot \phi_k +o(1)}{k+h_k-o(1)} = \phi_k+o(1) ~. \qedhere
\]
\end{proof}

\subsection{The limitation of static thresholds beyond \texorpdfstring{$k$}{k}-uniform matroids}

The strength of static threshold algorithms no longer holds for more structured matroids. We show that even for partition matroids, no static threshold can be constant-competitive. Note that a non-adaptive algorithm can still achieve competitive ratio at most $2$ for partition matroids: for each partition, set the optimal static threshold. However, a static threshold algorithm must use the \textbf{same} threshold for all partitions, which leads to super-constant competitive ratio in the worst case.

\begin{theorem}
\label{thm:static-partition-superconstant}
For every $m>1$, there exists a partition matroid prophet inequality instance with $m$ parts in which all variables are deterministic, such that every static threshold algorithm has competitive ratio $\Omega(\log m/\log\log m)$.
\end{theorem}

\begin{proof}
    Let $C>1$ be a parameter to be chosen later. The matroid has $m$ parts $B_1,\ldots,B_m$, each of rank $1$.  Part $B_i$ contains two elements $a_i$ and $b_i$, arriving in this order, with deterministic values $X_{a_i}=1/(Ci)$ and $X_{b_i}=1/i$.

    The prophet takes every $b_i$, so $\opt=\sum_{i=1}^m 1/i:=H_m$. Let $\alg_\tau$ be the value of a static threshold algorithm with threshold $\tau$. If $\tau=0$, then the algorithm obtains value $\sum_{i=1}^m 1/(Ci) = H_m/C$. Otherwise, for all $1\leq i<1/(C\tau)$, we have $X_{a_i}>\tau$, so the algorithm accepts $X_{a_i}$; for all $1/(C\tau)<i<1/\tau$, we have $X_{a_i}<\tau<X_{b_i}$, so the algorithm accepts $X_{b_i}$; for all $i>1/\tau$, we have $X_{a_i}<X_{b_i}<\tau$ and the algorithm obtains value $0$. The remaining values of $i$ involve tie-breaking, but we can simply upper bound their contribution by $O(1)$. Therefore
    \[
    \alg_\tau\le \sum_{i< 1/(C\tau)}\frac1{Ci}+\sum_{1/(C\tau)<i< 1/\tau}\frac1i+O(1)\le \frac{H_m}{C}+O(\log C) ~.
    \]
    Choosing $C=\lceil\log m\rceil$, we get $\alg_\tau=O(\log\log m)$. Hence every static threshold algorithm has competitive ratio $\Omega(\log m/\log\log m)$.
\end{proof}

The construction is similar to the one used in \citet{soda/BabaioffIK07} to show hardness for single-threshold algorithms in the matroid secretary setting, and we adapt it into the prophet inequality language. Since partition matroids are laminar matroids, the super-constant lower bound also extends to laminar matroids. We can also extend the construction to show super-constant lower bound for graphic matroids, even for simple graphs.

\paragraph{Extension to graphic matroids.} 
For general graphs, we can simply create a chain of $n$ pairs of parallel edges, and obtain an instance that is equivalent to the one used in the proof of Theorem~\ref{thm:static-partition-superconstant}. For simple graphs, we replace each pair of parallel edges with a triangle. In the $i$-th triangle, two edges have deterministic value $1/(Ci)$ and the other edge has deterministic value $1/i$. Then without tie-breaking, any static threshold algorithm can only get $2/(Ci)$ or $1/i$ from each triangle, and tie-breaking can only contribute an additional $O(1)$. The full analysis is similar to the proof of Theorem~\ref{thm:static-partition-superconstant} and is omitted.

\section{Deferred Proofs from Section~\ref{sec:laminar}}
\label{sec:appendix_laminar}

\subsection{Deferred Proofs from Section~\ref{sec:truncated_R1_partition_upper_bound}}
\label{sec:appendix_truncated_deferred_proofs}

We first restate the notations and provide some technical tools.

\begin{definition}[Restatement of Definition~\ref{def:zeta}]
For $\lambda\geq 0$, let $P_\lambda\sim\pois(\lambda)$, and define
\[
    \mu_k(\lambda):=\E[\min\{P_\lambda,k\}] ~.
\]
For $m\in[0,k)$, let $\lambda_k(m)$ be the unique value satisfying $\mu_k(\lambda_k(m))=m$, and set $\lambda_k(k)=\infty$. Define
\[
    \zeta_k(m):=\Pr[P_{\lambda_k(m)}\geq k] ~.
\]
Equivalently, if $\mu_k(\lambda)=m$, then $\zeta_k(m)=\Pr[P_\lambda\geq k]$.
\end{definition}

\begin{lemma}
\label{lem:zeta_property}
For every $k\geq 1$, the function $\mu_k$ is continuous and strictly increasing from $0$ to $k$. Moreover, for $m=\mu_k(\lambda)\in(0,k)$, we have $\zeta_k'(m) \leq 1$, and $\zeta_k(m)-m$ is nonincreasing in $m$.
\end{lemma}
\begin{proof}
The continuity and strict monotonicity of $\mu_k$ follow from the standard identity $\mu_k'(\lambda) = \Pr[P_\lambda \leq k-1] > 0$.
Moreover,
$\mu_k(0)=0$ and $\mu_k(\lambda)\to k$ as $\lambda\to\infty$.

Next, for $m=\mu_k(\lambda)\in(0,k)$, since $\zeta_k(m) = \Pr[P_{\lambda}\geq k]$, by the chain rule, 
\[
    \zeta_k'(m) = (\Pr[P_\lambda\geq k])' \cdot \lambda_k'(m) ~.
\]
The first term is equal to $\Pr[P_{\lambda}=k-1]$, and the second term is equal to $\frac{1}{\mu_k'(\lambda)} = \frac{1}{\Pr[P_\lambda \leq k-1]}$ since $\lambda_k$ is the inverse function of $\mu_k$. So we conclude that $\zeta_k'(m) = \frac{\Pr[P_\lambda=k-1]}{\Pr[P_\lambda\leq k-1]}$. Since $\Pr[P_\lambda=k-1]\leq \Pr[P_\lambda\leq k-1]$ for all $\lambda\geq 0$, we have
$\zeta_k'(m)\leq 1$. 

Therefore, $(\zeta_k(m)-m)' =\zeta_k'(m)-1\leq 0$, which implies $\zeta_k(m)-m$ is nonincreasing.
\end{proof}

\begin{lemma}
\label{lem:Poisson binomial-log-concavity}
Let $S$ be a Poisson binomial random variable and write $s_t:=\Pr[S=t]$ and $L_t:=\Pr[S\leq t]$. Then the sequence $(s_t)$ is log-concave. In particular, this implies $s_t/L_t$ is nonincreasing in $t$.
\end{lemma}
This follows from Hoggar's theorem that convolution preserves log-concavity for independent integer-valued random variables; see, e.g., \cite{hoggar1974chromatic,johnson2006preservation}.

\begin{lemma}[Restatement of Lemma~\ref{lem:capped_poisson_binomial}]
\label{lem:capped_poisson_binomial_restate}
    Let $X$ be a Poisson binomial random variable, and suppose a random variable $Z$ and an event $\mathcal E$ satisfy $\mathcal E \subseteq \{X\geq k\}$ and $Z=\min\{X,k-1\}+\mathbf 1_{\mathcal E}$, then $\Pr[\mathcal E]\leq \zeta_k(\E[Z])$.
\end{lemma}

The proof is a corollary of the following more general statement.

\begin{lemma}
\label{lem:bernoulli_splitting}
Let $X$ be a Poisson binomial random variable, define
\[
    M_k(X):=\E[\min\{X,k\}],\qquad T_k(X):=\Pr[X\geq k].
\] 
Then $T_k(X)\leq \zeta_k(M_k(X))$ for all integer $k\geq 2$.
\end{lemma}

\begin{proofof}{Lemma~\ref{lem:capped_poisson_binomial_restate}}
Let $M=M_k(X)$, then since $\mathcal{E}\subseteq \{X\geq k\}$ and $Z=\min\{X,k-1\}+\mathbf 1_{\mathcal E}$, we have $\E[Z]\leq M$. Further, taking expectations in $1_{\mathcal E} = Z-\min\{X,k-1\}$,
\[
    \Pr[\mathcal E]=\E[Z]-\E[\min\{X,k-1\}] ~.
\]
Since $\E[\min\{X,k-1\}]=M-T_k(X)$, by Lemma~\ref{lem:bernoulli_splitting} we get
\[
    \Pr[\mathcal E]=\E[Z]-M+T_k(X)\leq \E[Z]-M+\zeta_k(M) \leq \zeta_k(\E[Z]) ~.
\]
The last inequality follows from Lemma~\ref{lem:zeta_property} and $\E[Z]\leq M$.
\end{proofof}

\begin{proofof}{Lemma~\ref{lem:bernoulli_splitting}}
For the corner case $M_k(X)=k$ and $T_k(X)=1$, the claim is trivial. Throughout the rest of the proof we assume $M_k(X)<k$ and $T_k(X)<1$.

At a high level, the proof first shows that one can replace a Bernoulli trial by two smaller independent Bernoulli trials so that the capped mean $M_k(X)=\E[\min\{X,k\}]$ is preserved and the tail $T_k(X)=\Pr[X\geq k]$ weakly increases. Iterating this splitting operation gives a sequence of Poisson binomial variables whose maximum Bernoulli probability tends to $0$, while preserving $M_k$ and weakly increasing $T_k$. Therefore, for fixed $M=M_k(X)$, the largest possible value of $T_k(X)$ is achieved in the Poisson limit, namely by the Poisson variable $P_\lambda$ satisfying $\E[\min\{P_\lambda,k\}]=M$.

We start with proving the one-step replacement. Decompose $X$ into $X=S+\ber(p)$, where $S$ can be viewed as a Poisson binomial random variable independent of $\ber(p)$. Write
\[
    L_t:=\Pr[S\leq t], \qquad s_t:=\Pr[S=t] ~.
\]
By the assumption $T_k(X)<1$, we have $L_{k-1}>0$. $\ber(p)$ contributes to $M_k(X)$ when $S\leq k-1$, so 
\[
    M_k(X) = \E[\min\{S,k\}] + p L_{k-1} ~.
\]
Replacing $\ber(p)$ by two independent copies of $\ber(q)$, i.e., let $X' = X-\ber(p) + 2\ber(q)$, then
\[
    M_k(X') = \E[\min\{S,k\}] + 2q L_{k-1}-q^2 s_{k-1} ~.
\]
We choose $q\in[0,1]$ so that $M_k(X') = M_k(X)$, i.e., $2qL_{k-1}-q^2s_{k-1}=pL_{k-1}$.

Now we compare $T_k(X)$ and $T_k(X')$. We have
\[
T_k(X) = \Pr[S\geq k] + ps_{k-1} ~,
\]
and
\[
T_k(X') = \Pr[S\geq k] + (2q-q^2)s_{k-1} + q^2 s_{k-2} ~.
\]
Therefore, plugging in $p = 2q - q^2 s_{k-1}/L_{k-1}$, we get
\[
    T_k(X')-T_k(X) = q^2\left(\frac{s_{k-1}^2}{L_{k-1}}+s_{k-2}-s_{k-1}\right) ~.
\]
It remains to show that the term inside the parentheses is nonnegative, which is equivalent to
\[
    s^2_{k-1} + s_{k-2}L_{k-1} - s_{k-1}(L_{k-2}+s_{k-1}) =s_{k-2}L_{k-1}-s_{k-1}L_{k-2} \geq 0 ~.
\]
This follows from Lemma~\ref{lem:Poisson binomial-log-concavity}.

Given the one-step replacement argument, we repeatedly apply it to every Bernoulli trial. Let $X_r$ be the resulting Poisson binomial variable after $r$ steps of replacement. Then for every $r$,
\[
    M_k(X_r)=M_k(X), \qquad T_k(X_r)\geq T_k(X) ~.
\]
Moreover, the maximum Bernoulli probability in $X_r$ tends to $0$. Indeed, from $2qL_{k-1}-q^2s_{k-1}=pL_{k-1}$ and $s_{k-1}\leq L_{k-1}$, we have $p\geq 2q-q^2$, and hence $q\leq 1-\sqrt{1-p}<p$ whenever $0<p<1$. If $p=1$, then $q=1$ can only happen when $L_{k-2}=0, L_{k-1}=1$, in which case $S\geq k-1$ holds almost surely, so $X=S+\ber(1) \geq k$ holds almost surely, contradicting our assumption $M_k(X)<k$.

Next we show that the means $\E[X_r]$ remain bounded. Otherwise, along a subsequence we have $\E[X_r]\to\infty$. Since $X_r$ is Poisson binomial, $\operatorname{Var}(X_r)\leq \E[X_r]$, so Chebyshev's inequality gives $\Pr[X_r<k]\to0$. This implies $M_k(X_r)\to k$, contradicting $M_k(X_r)=M_k(X)<k$.

Thus, along some subsequence, $\E[X_r]\to\lambda<\infty$. Since the maximum Bernoulli probability tends to $0$, the law of small numbers implies that $X_r$ converges in distribution to $P_\lambda$ along this subsequence. Therefore $M_k(X)=M_k(P_\lambda)$. We conclude that
\[
    T_k(X)\leq \lim_r T_k(X_r)=T_k(P_\lambda)=\zeta_k(M_k(X)) ~. \qedhere
\]
\end{proofof}

\begin{lemma}[Restatement of Lemma~\ref{lem:alpha_tight_all_k}]
\label{lem:alpha_tight_all_k_restate}
Let $R(r)=\frac{2+r}{2r+(1-r)\exp(-r/(1-r))}$ be the lower bound ratio from Theorem~\ref{thm:truncated_lower_bound}, and $r^*$ be the unique maximizer of $R(r)$ on $(0,1)$. Set $\alpha^*:=1/R(r^*)$, then for every $k\geq 2$ and every $a\in[0,1]$, we have
\[
    G_k(a) = 1-\alpha^*a-\zeta_k(\alpha^*(k-a))\geq \alpha^*.
\]
\end{lemma}

\begin{proof}
Fix $k\geq 2$ and take derivative of $G_k(a)$:
\[
G_k'(a) = -\alpha^* + \alpha^*\zeta_k'(\alpha^*(k-a)) \leq 0 ~.
\]
The inequality follows from $0<\alpha^*(k-a)<k$ and $\zeta_k'(m)\leq 1,\forall m\in(0,k)$. Thus, $G_k(a)$ is minimized at $a=1$ over $a\in[0,1]$, and it remains to prove $G_k(1)\geq \alpha^*$, i.e.,
\[
1-2\alpha^* \geq \zeta_k(\alpha^*(k-1)) ~.
\]

\paragraph{Case 1: $k=2$.} We first prove that for $k=2$ the inequality is tight: $1-2\alpha^* = \zeta_2(\alpha^*)$. By standard properties of Poisson variables, we have
\[
    \Pr[P_\lambda\ge2]=1-(1+\lambda)e^{-\lambda}, \qquad \mu_2(\lambda)=\E[\min\{P_\lambda,2\}]=2-(2+\lambda)e^{-\lambda} ~.
\]
Let $\lambda^*>0$ be the unique value satisfying $1-2\mu_2(\lambda^*) = \zeta_2(\mu_2(\lambda^*))= \Pr[P_{\lambda^*}\geq 2]$, i.e.,
\[
    2(2+\lambda^*)e^{-\lambda^*}-3 = 1-(1+\lambda^*)e^{-\lambda^*}\iff (5+3\lambda^*)e^{-\lambda^*}=4 ~.
\]
Next we show that $\alpha^*$ is equal to $\mu_2(\lambda^*)=2-(2+\lambda^*)e^{-\lambda^*}$, so by construction $\zeta_2(\alpha^*)=1-2\alpha^*$.

Recall that the first-order optimality condition of $R(r)$ is
\[
    (5-2r)\exp\left(-\frac{r}{1-r}\right)=4(1-r) ~.
\]
Substituting $\lambda=\frac{r}{1-r}$, we again obtain $(5+3\lambda)e^{-\lambda} = 4$, so if $r^*$ is the maximizer of $R(r)$, then $\lambda^* = \frac{r^*}{1-r^*}$. Therefore, we conclude that
\[
\alpha^* = \frac{1}{R(r^*)} = \frac{2r^*+(1-r^*)\exp(-r^*/(1-r^*))}{2+r^*} = \frac{2\lambda^*+e^{-\lambda^*}}{2+3\lambda^*} = 2-(2+\lambda^*)e^{-\lambda^*} ~.
\] 
The last equality follows from $(5+3\lambda^*)e^{-\lambda^*}=4$. This explains why the upper bound and the lower bound meet at exactly the same value.

\paragraph{Case 2: $k\geq 3$.} Let $\lambda_0:=\frac{k-1}{2}$ and $P_0=P_{\lambda_0}$. Define $\tau_k:=\Pr[P_0\geq k]$. We first claim that $\tau_k\leq 1-2\alpha^*$, assuming it is correct for the moment. 

Since $\min\{P_0,k\}=P_0-(P_0-k)^+$, we have
\[
    \mu_k(\lambda_0) = \E[\min\{P_0,k\}] = \lambda_0-\E[(P_0-k)^+] ~.
\]
Moreover, using the standard property $\E[P_0\mathbf 1\{P_0\geq k+1\}]=\lambda_0\Pr[P_0\geq k]$, we have
\[
    \E[(P_0-k)^+] = \E[P_0\mathbf 1\{P_0\geq k+1\}] - k\Pr[P_0\geq k+1] \leq \lambda_0\Pr[P_0\geq k]= \lambda_0\tau_k ~.
\]
Thus $\mu_k(\lambda_0)\geq \lambda_0(1-\tau_k)$. Using $\tau_k\leq 1-2\alpha^*$, we get
\[
    \alpha^*(k-1) = \frac{k-1}{2}\cdot 2\alpha^* \leq \lambda_0(1-\tau_k) \leq \mu_k(\lambda_0) ~,
\]
which implies $\zeta_k(\alpha^*(k-1)) \leq \zeta_k(\mu_k(\lambda_0)) = \tau_k \leq 1-2\alpha^*$.

It remains to prove the claim $\tau_k\leq 1-2\alpha^*$. For every $j\geq k$, we have
\[
    \frac{\Pr[P_0=j+1]}{\Pr[P_0=j]}=\frac{\lambda_0}{j+1}\leq \frac{\lambda_0}{k+1} ~.
\]
Therefore we can upper bound $\tau_k$ by
\begin{align*}
    \tau_k = &~ \sum_{j=k}^{\infty}\Pr[P_0=j] \\
    \leq &~ \Pr[P_0=k]\sum_{j=0}^{\infty}\left(\frac{\lambda_0}{k+1}\right)^j \\
    = &~ \frac{e^{-\lambda_0 }\lambda_0^k/k!}{1-\frac{\lambda_0}{k+1}} =  \frac{e^{-\frac{k-1}{2}}\left(\frac{k-1}{2}\right)^k/k!}{\frac{k+3}{2(k+1)}} := U_k~.
\end{align*}
Observe that
\[
    \frac{U_{k+1}}{U_k} = 
    \frac{e^{-1/2}}{2}\left(\frac{k}{k-1}\right)^k\frac{k(k+2)(k+3)}{(k+1)^2(k+4)}<1 ~.
\]
The last inequality is verified as follows: for $k=3$, direct calculation gives $\frac{U_4}{U_3} = \frac{405}{448\sqrt e}<1$; for $k\geq 4$, we use $\left(\frac{k}{k-1}\right)^k\leq \left(\frac{4}{3}\right)^4$ and $\frac{k(k+2)(k+3)}{(k+1)^2(k+4)}<1$, then $\frac{U_{k+1}}{U_k} < \frac{e^{-1/2}}{2}\left(\frac{4}{3}\right)^4 = \frac{128}{81\sqrt e} < 1$.
Therefore, for all $k\geq 3$ we have
\[
    \tau_k \leq U_k\leq U_3 = \frac{e^{-1}/3!}{6/8} =\frac{2}{9e}<1-2\alpha^* ~. \qedhere
\]
\end{proof}

\subsection{A stronger lower bound for laminar matroids}
\label{sec:appendix_laminar_lower_bound}

The construction recursively applies the idea of the basic $2.179$ lower bound for truncated $R_1$-partition matroids: starting from the two-variable prophet inequality hard instance, iteratively add a large number of Bernoulli variables and form a new level of capacity constraint.

\begin{theorem}[Restatement of Theorem~\ref{thm:laminar_lower_bound}]
\label{thm:laminar_lower_bound_restate}
There exists a laminar matroid prophet inequality instance such that every non-adaptive algorithm has competitive ratio at least $2.217$.
\end{theorem}

\begin{proof}
For a laminar matroid prophet inequality instance $L$, let $\opt(L)$ denote the prophet value and $\alg(L)$ denote the expected value of the best non-adaptive algorithm. 

\paragraph{Base instance.}
The base instance $L_0$ is the classical two-variable prophet inequality hard instance with two elements $g_1,g_2$, where $X_{g_1}=1$, and $X_{g_2}=1/\epsilon$ with probability $\epsilon$, and $0$ otherwise. There is one constraint $|I \cap \{g_1,g_2\}|\leq 1$. $g_1$ arrives before $g_2$, so as $\epsilon \to 0$, we have $\opt(L_0) \to 2, \alg(L_0)\to 1$.

\paragraph{Recursive construction.}
Take a copy of $L_{\ell-1}$ and let $E_{\ell-1}$ denote the set of items. To construct $L_\ell$, add $N_\ell$ items $B_\ell$ with identical distributions, each taking value $r_\ell$ with probability $b_\ell$ and $0$ otherwise. All previous laminar constraints inside $L_{\ell-1}$ are preserved, and there is one additional constraint that $|I\cap (B_\ell\cup E_{\ell-1})|\leq \ell+1$. The arrival order is: first $B_\ell$ arrives, then all elements in $E_{\ell-1}$ arrive in the same order as in $L_{\ell-1}$. We choose the parameters such that $b_\ell \to 0, N_\ell b_\ell \to \infty$, and $r_\ell<r_{\ell-1}$ will be optimized later.

\paragraph{Prophet value.}
Since $N_\ell b_\ell \to \infty$, for each block $B_\ell$ at least one value $r_\ell$ is realized with probability tending to $1$. If that happens, since $r_\ell<r_{\ell-1}$, the prophet should take exactly one value $r_\ell$ from each $B_\ell$. Therefore, we have $\opt(L_\ell) \to 2 + \sum_{i=1}^\ell r_i$.

\paragraph{Non-adaptive algorithms.}
For the analysis, let $\alg_\ell(s)$ denote an upper bound on the expected value of the best non-adaptive algorithm for $L_\ell$ when at most $s$ items can be selected from $L_\ell$. For $B_\ell$, essentially a non-adaptive algorithm only needs to decide for how many variables in $B_\ell$ to set a threshold $r_\ell$. By symmetry, suppose it sets threshold $r_\ell$ for the first $k$ items in $B_\ell$, and let $Z\sim\operatorname{Bin}(k,b_\ell)$ be the number of items among the first $k$ that realize value $r_\ell$. Then the algorithm obtains value $r_\ell\min\{Z,s\}$ from $B_\ell$, and has remaining capacity $\min\{\ell,(s-Z)^+\}$ for the subsequent copy of $L_{\ell-1}$. Therefore, the expected value of the algorithm on the whole instance is at most
\[
r_\ell\E[\min\{Z,s\}] + \E\left[\alg_{\ell-1}(\min\{\ell,(s-Z)^+\})\right] ~.
\]
Since we take $b_\ell\to0$, if $kb_\ell\to\lambda$, then $Z$ converges to $Z_\lambda\sim\pois(\lambda)$. The best non-adaptive algorithm picks $k$ that maximizes the expected value, thus we obtain the recurrence
\[
\alg_\ell(s) := \sup_{\lambda\geq 0}\left\{r_\ell\E[\min\{Z_\lambda,s\}] + \E\left[\alg_{\ell-1}(\min\{\ell,(s-Z_\lambda)^+\})\right]\right\} ~.
\]
Finally, by definition we have $\alg(L_\ell) \leq  \alg_\ell(\ell+1)$.

Therefore, given any set of values $r_1>r_2>\ldots>r_\ell$, we can compute $\opt(L_\ell)$ and recursively compute a good estimate on $\alg_\ell(\ell+1)$. We optimize the parameters for $\ell=3$, and at $(r_1,r_2,r_3)=(0.34,0.10,0.02)$, we have $\opt(L_3)\to 2.46, \alg(L_3)\leq 1.1092$, yielding competitive ratio at least $2.217$. The lower bound is certified in the sense that, for each $\alg_\ell(s)$, we evaluate an upper bound of the supremum over $\lambda$ by estimating on a grid and adding a carefully bounded error term \footnote{Details and code of the numeric optimization can be found at \url{https://github.com/relyt871/Numeric-Optimization-of-Bounds-for-Laminar-Matroids}.}. 
Larger values of $\ell$ can in principle further improve the lower bound, but the marginal improvement is very small, and the error in the estimate will accumulate and prevent us from certifying a better lower bound.
\end{proof}

\subsection{Improved upper bounds for truncated partition and laminar matroids}
\label{sec:appendix_laminar_upper_bound}

\begin{proofof}{Theorem~\ref{thm:truncated_upper_bound}}
Let $k$ denote the global rank. We apply the non-adaptive OCRS framework in Section~\ref{sec:truncated_R1_partition_upper_bound} with $\alpha=1/3$. It suffices to prove inductively that $F_i\ge\alpha$ for every item $i$.

Fix an item $i$ and assume $F_j\geq \alpha$ for all $j<i$, so each previous item $j$ is accepted with probability $\alpha x_j$.  Suppose $i$ belongs to partition $B$, and let $r_B$ denote the local rank of $B$. Since $x\in P(\mathcal M)$, we have
\[
    \sum_{j<i} x_j \leq k ~, \qquad \sum_{j\in B, j<i}x_j \leq r_B ~.
\]

Let $A$ be the set of previously accepted items before $i$ arrives. Item $i$ is infeasible only if either the global rank is full, i.e., $|A|=k$, or its local partition is full, i.e., $|A\cap B|=r_B$. By inductive hypothesis, every previous item $j$ has acceptance marginal $\alpha x_j$, therefore
\[
    \E[|A|]\leq \alpha k ~, \qquad  \E[|A\cap B|]\le \alpha r_B ~.
\]
By Markov's inequality and $|A|\leq k, |A\cap B|\leq r_B$, this implies
\[
    \Pr[|A|=k] \leq \alpha ~, \qquad \Pr[|A\cap B|=r_B] \leq \alpha ~.
\]
Thus by union bound, item $i$ is infeasible with probability at most $2\alpha$, hence $F_i \geq 1-2\alpha = \alpha$.
\end{proofof}

We conjecture that the tight proof of Theorem~\ref{thm:truncated-R1-partition} can be extended to truncated partition matroids and also obtain the tight ratio of $\approx 2.179$, but the extension seems challenging. When the local partitions can have ranks larger than $1$, the event that the global constraint becomes tight before item $i$ arrives can no longer be cleanly related to a Poisson binomial variable. Further, the items in the same partition as $i$ can also interfere with the global constraint, so we cannot decompose the event that $i$ becomes infeasible into two disjoint events and consider external partitions and internal items separately.

\paragraph{Computing $F_i$ efficiently.} For truncated partition matroids, we can compute $F_i$ efficiently. Suppose item $i$ belongs to partition $B$, then it remains feasible if and only if less than $k$ previous items have been accepted, and less than $r_B$ previous items in partition $B$ have been accepted. Formally, for each partition $C$, let $Y_C$ be a random variable denoting the number of previous items in $C$ that realize a value above the corresponding threshold, then $i$ is feasible if and only if $Y_B<r_B$ and $Y_B+\sum_{C\neq B}\min\{Y_C,r_C\} < k$.

When processing item $i$, since we have already computed the thresholds for the first $i-1$ items, we are able to explicitly compute the distribution of each $Y_C$. For each partition $C$ and $t\in\{0\}\cup [r_C]$, let $p_{C,t}$ denote the probability that $\min\{Y_C,r_C\}=t$, which can be computed via a simple knapsack dynamic programming. Suppose the partitions other than $B$ are labeled $C_1,C_2,\ldots,C_m$, then $F_i$ can be expressed via a convolution over the partitions: 
\[
F_i = \sum_{\ell=0}^{r_B-1}p_{B,\ell}\cdot \sum_{s=0}^{k-\ell-1}\left(\sum_{0\leq t_j\leq r_{C_j},\forall j\in[m], \sum_{j=1}^m t_j=s}\prod_{j=1}^m p_{C_j,t_j}\right) ~.
\]
This convolution can be evaluated in $\mathrm{poly}(i,k)$ time using a brute force dynamic programming, and can be further accelerated using standard divide-and-conquer techniques.

\begin{proofof}{Theorem~\ref{thm:laminar_upper_bound}}
For laminar matroids, $F_i$ depends on all constraints that involve $i$, and the analysis becomes very complicated. We analyze a simpler and weaker version of the non-adaptive OCRS procedure: fix $\beta\in[0,1]$, then for each item $i$, set its threshold to $\tau_i(\beta x_i)$. We say that $i$ is active if it realizes a value above this threshold.

Fix an item $i$. For each positive integer $r$, let $S_r$ be the largest laminar constraint of rank $r$ that contains $i$, if such a constraint exists. Then $i$ is feasible as long as for every $r$ such that $S_r$ exists, at most $r-1$ previous items in $S_r$ are active.

Let $Y_r$ be the random variable indicating the number of previous active items in $S_r$. Since $x\in P(\mathcal M)$, we have $x(S_r)\leq r$, so $\E[Y_r] \leq \beta r$. Therefore, by Lemma~\ref{lem:bernoulli_splitting},
\[
\Pr[Y_r\geq r] \leq \zeta_r(\E[\min\{Y_r,r\}]) \leq \zeta_r(\E[Y_r]) \leq \zeta_r(\beta r) ~.
\]

Let $E_r$ denote the event that at most $r-1$ previous items in $S_r$ are active, then the above implies $\Pr[E_r] \geq 1-\zeta_r(\beta r)$. The events $E_r$ are decreasing events in the indicators of whether each previous item is active, so we can apply the FKG inequality~\cite{fkg},
\[
    \Pr\left[\bigcap_{r\geq 1} E_r\right] \geq \prod_{r\geq 1} \Pr[E_r] \geq \prod_{r\geq 1}\left(1-\zeta_r(\beta r)\right) := \Gamma(\beta) ~.
\]
Therefore, item $i$ is feasible with probability at least $\Gamma(\beta)$, hence it is accepted with probability at least $\beta\Gamma(\beta)\cdot x_i$. The ex-ante relaxation implies
\[
    \E[\alg]\ge \beta \Gamma(\beta)\cdot \sum_i \nu_i(x_i)\ge \beta \Gamma(\beta)\cdot\E[\opt] ~,
\]
so the competitive ratio is $\frac{1}{\beta\Gamma(\beta)}$. 
Numerically, $\min_{0<\beta<1}\frac{1}{\beta\Gamma(\beta)}\approx 6.311$, which is attained at $\beta\approx0.2788$ \footnote{Details and code of the numeric optimization can be found at \url{https://github.com/relyt871/Numeric-Optimization-of-Bounds-for-Laminar-Matroids}.}. 
\end{proofof}

\section{Deferred Proofs from Section~\ref{sec:graphic}}
\label{sec:appendix_graphic}

\subsection{Improved lower bound via recursive hats}
\label{sec:appendix_graphic_recursive}

We now recursively use the one-layer construction to obtain a stronger lower bound. The key observation is that the two parallel brim edges in the one-layer construction play the same role as the pair \(e_{i,1},e_{i,2}\) in a hat.
Thus, in the next layer, we replace this pair by a whole copy of the one-layer hard instance, while keeping $e_{i,3}$ untouched. This operation can be repeated iteratively, and can amplify the lower bound from $2.5$ to arbitrarily close to $3$. Figure~\ref{fig:two-layer-hat} demonstrates the recursive construction for two layers.

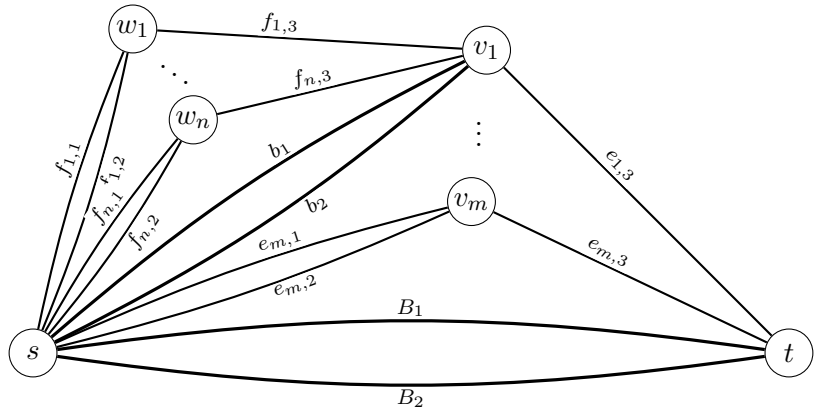
\begin{figure}[ht]
\centering
\begin{tikzpicture}[
    x=1cm, y=1cm,
    vertex/.style={circle, draw, fill=white, inner sep=1.2pt, minimum size=18pt},
    edge/.style={draw, thick},
    brim/.style={draw, very thick},
    elabel/.style={font=\scriptsize, fill=white, inner sep=1pt}
]

\node[vertex] (s) at (-5,0) {$s$};
\node[vertex] (t) at (5,0) {$t$};

\draw[brim, bend left=8] (s) to node[elabel, above] {$B_1$} (t);
\draw[brim, bend right=8] (s) to node[elabel, below] {$B_2$} (t);

\node[vertex] (v1) at (1.0,4) {$v_1$};
\node[vertex] (vm) at (0.8,2) {$v_m$};

\node at (0.9,3) {$\vdots$};

\draw[edge] (s) to[bend left=5]
    node[elabel, pos=0.6, above, sloped] {$e_{m,1}$} (vm);
\draw[edge] (s) to[bend right=5]
    node[elabel, pos=0.6, below, sloped] {$e_{m,2}$} (vm);
\draw[edge] (t) --
    node[elabel, pos=0.6, above, sloped] {$e_{m,3}$} (vm);

\draw[edge] (t) --
    node[elabel, pos=0.6, above, sloped] {$e_{1,3}$} (v1);


\coordinate (m) at ($(s)!0.52!(v1)$);

\node[vertex] (w1) at ($(m)+(-1.8,2.2)$) {$w_1$};
\node[vertex] (wn) at ($(m)+(-1.0,1.0)$) {$w_n$};

\node at ($(m)+(-1.25,1.7)$) {$\ddots$};

\draw[brim, bend left=7] (s) to node[elabel, pos=0.6, above, sloped] {$b_1$} (v1);
\draw[brim, bend right=7] (s) to node[elabel, pos=0.6, below, sloped] {$b_2$} (v1);

\draw[edge] (s) to[bend left=5]
    node[elabel, pos=0.6, above, sloped] {$f_{1,1}$} (w1);
\draw[edge] (s) to[bend right=5]
    node[elabel, pos=0.6, below, sloped] {$f_{1,2}$} (w1);
\draw[edge] (v1) --
    node[elabel, pos=0.6, above, sloped] {$f_{1,3}$} (w1);

\draw[edge] (s) to[bend left=5]
    node[elabel, pos=0.6, above, sloped] {$f_{n,1}$} (wn);
\draw[edge] (s) to[bend right=5]
    node[elabel, pos=0.6, below, sloped] {$f_{n,2}$} (wn);
\draw[edge] (v1) --
    node[elabel, pos=0.6, above, sloped] {$f_{n,3}$} (wn);

\end{tikzpicture}
\caption{Two-layer construction. For convenience, we only demonstrate the recursive hat for $(s,v_1)$.}
\label{fig:two-layer-hat}
\end{figure}

\begin{theorem}[Restatement of Theorem~\ref{thm:graphic_lower_bound_3}]
\label{thm:graphic_lower_bound_restate}
    For every integer $k\geq 1$, there exists a family of graphic matroid prophet inequality instances such that every non-adaptive algorithm has competitive ratio at least $3-\frac{1}{k+1}$.
\end{theorem}

\begin{proof}
Let $n$ be a sufficiently large integer. Define $p_0=n^{-1/3}, H_0=1$, and for all $\ell\geq 1$, define
\[
p_\ell=p_{\ell-1}^2,\qquad N_\ell=p_{\ell-1}^{-3},\qquad H_\ell=\prod_{t=1}^{\ell}N_t ~.
\]
Then
\[
p_\ell=n^{-2^\ell/3},\qquad N_\ell=n^{2^{\ell-1}},\qquad H_\ell=n^{2^\ell-1} ~.
\]
In particular, $p_\ell\to0$ and $N_\ell p_\ell=p_{\ell-1}^{-1}\to\infty$ for every fixed $\ell$.

\paragraph{Base case.}
The base gadget $G_0(s,t)$ simply embeds the classic two-variable prophet inequality instance: there are two parallel edges $e_1,e_2$ between $s$ and $t$, where $e_1$ has deterministic value $1$, and $e_2$ takes value $1/p_0$ with probability $p_0$ and $0$ otherwise. The edge $e_1$ arrives before $e_2$.

\paragraph{Recursive construction.}
For each $\ell\geq 1$, $G_\ell(s,t)$ is constructed as follows: create vertices $v_1,\ldots,v_{N_\ell}$. For each $j\in[N_\ell]$, construct a recursive hat consisting of a copy of $G_{\ell-1}(s,v_j)$, and an edge $r_j=(v_j,t)$ that takes value $H_{\ell-1}/p_{\ell-1}$ with probability $p_{\ell-1}$ and $0$ otherwise. Finally, add two brim edges $b_1,b_2$ between $s$ and $t$: $b_1$ has deterministic value $H_\ell$, $b_2$ takes value $H_\ell/p_\ell$ with probability $p_\ell$ and $0$ otherwise. 

The arrival order is: sequentially for each $j=1,\ldots,N_\ell$, the copy $G_{\ell-1}(s,v_j)$ arrives in its original order, then $r_j$ arrives; finally $b_1,b_2$ arrive one after another.

\paragraph{Prophet value.}
Let $\opt_\ell$ denote the prophet value from $G_\ell(s,t)$. We prove by induction that $\opt_\ell\geq (3\ell+2)H_\ell-o(H_\ell)$. In the base case the prophet value is indeed $2-o(1) = 2H_0-o(H_0)$.

For $\ell\geq 1$, the prophet always accepts the larger among the two final brim edges and obtains expected value $(2-o(1))H_\ell$. For each recursive hat, if $r_j$ realizes value $0$, the prophet obtains value $\opt_{\ell-1}$ from $G_{\ell-1}(s,v_j)$; otherwise, the prophet accepts $r_j$ with value $H_{\ell-1}/p_{\ell-1}$. Therefore, assuming the inductive hypothesis $\opt_{\ell-1} \geq (3\ell-1)H_{\ell-1} - o(H_{\ell-1})$, each recursive hat contributes at least
\[
    (1-p_{\ell-1})\opt_{\ell-1} + p_{\ell-1}\cdot H_{\ell-1}/p_{\ell-1} = 3\ell\cdot H_{\ell-1}-o(H_{\ell-1}) ~.
\]
Summing over the brim edges and all $N_\ell$ recursive hats gives
\[
    \opt_\ell \geq N_\ell\cdot (3\ell H_{\ell-1}-o(H_{\ell-1})) + (2-o(1))H_\ell = (3\ell+2)H_\ell - o(H_\ell) ~.
\]

\paragraph{Non-adaptive algorithm.}
For each $\ell$, let $\alg_\ell = (\ell+1)H_\ell, \alg'_\ell = \ell H_\ell$, then we can prove the following by induction on $\ell$: 
\begin{enumerate}
    \item every non-adaptive algorithm on $G_\ell(s,t)$ that never connects $s,t$ has expected value at most $\alg'_\ell$. 
    \item every non-adaptive algorithm on $G_\ell(s,t)$ has expected value at most $\alg_\ell$, and to obtain value larger than $\alg'_\ell$, it connects $s,t$ with probability at least $p_\ell$.
\end{enumerate}

For the base case $\ell=0$, every non-adaptive algorithm has expected value at most $1$, and every non-adaptive algorithm with positive expected value connects $s,t$ with probability at least $p_0$, which corresponds to ignoring $e_1$ and setting threshold $1/p_0$ for $e_2$.

For $\ell\geq 1$, suppose the two claims hold for $\ell-1$. 

First consider algorithms on $G_\ell(s,t)$ that never connect $s,t$. The two brim edges must be ignored. For each recursive hat $j$, if the algorithm ignores $r_j$, then there are no constraints on $G_{\ell-1}(s,v_j)$ and the algorithm can obtain at most $\alg_{\ell-1}=\ell H_{\ell-1}$ from $G_{\ell-1}(s,v_j)$; if the algorithm sets a threshold $H_{\ell-1}/p_{\ell-1}$ for $r_j$ and obtains expected value $H_{\ell-1}$ from $r_j$, then in order to avoid connecting $s,t$, it must guarantee that $s,v_j$ is never connected, so it can only obtain value at most $\alg'_{\ell-1} = (\ell-1)H_{\ell-1}$ from $G_{\ell-1}(s,v_j)$. Therefore, for each recursive hat the algorithm gets value at most $\ell H_{\ell-1}$. Thus, we conclude that any algorithm on $G_\ell(s,t)$ that never connects $s,t$ obtains expected value at most $N_\ell \cdot \ell H_{\ell-1} = \ell H_\ell$.

Next we upper bound the value of an arbitrary non-adaptive algorithm. For a recursive hat $j$, as argued above, if the algorithm guarantees that $s,t$ is never connected through hat $j$ then it obtains value at most $\ell H_{\ell-1}$ from the hat. We call such hats ``safe''.

Otherwise, the algorithm sets threshold $H_{\ell-1}/p_{\ell-1}$ for $r_j$ and also obtains value larger than $\alg'_{\ell-1}$ from $G_{\ell-1}(s,v_j)$. By inductive hypothesis, the expected value from hat $j$ is at most $\alg_{\ell-1}+p_{\ell-1}\cdot H_{\ell-1}/p_{\ell-1} = (\ell+1)H_{\ell-1}$, and $s,v_j$ is connected with probability at least $p_{\ell-1}$, so $s,t$ is connected with probability at least $p_{\ell-1}^2 = p_\ell$. We call such hats ``risky''.

Suppose the algorithm uses $m$ risky hats. Since the recursive hats are identical, we may assume they are the first $m$ hats. Let $q_\ell=1-p_\ell$. Before the $j$-th risky hat arrives, the probability that $s,t$ are still disconnected is at most $q_\ell^{j-1}$. Further, if $s,t$ are already connected when a risky hat $j$ arrives, then with probability at least $p_{\ell-1}$, $s,v_j$ will become connected, and then the algorithm can no longer accept $r_j$, so in this case the expected value from hat $j$ is at most $(\ell+1-p_{\ell-1})H_{\ell-1}$. Therefore, the expected value from the recursive hats is at most
\begin{align*}
    &~ (N_\ell-m)\ell H_{\ell-1} + \sum_{j=1}^m\left(q_\ell^{j-1}(\ell+1)H_{\ell-1} + (1-q_\ell^{j-1})(\ell+1-p_{\ell-1})H_{\ell-1}\right)\\
    = &~ (N_\ell-m)\ell H_{\ell-1} + \sum_{j=1}^m\left((\ell+1-p_{\ell-1})H_{\ell-1} +q_\ell^{j-1}p_{\ell-1}H_{\ell-1}\right)\\
    = &~ \ell N_\ell H_{\ell-1} - m\ell H_{\ell-1} + m(\ell+1-p_{\ell-1})H_{\ell-1} + p_{\ell-1}H_{\ell-1}\cdot \frac{1-q_\ell^m}{1-q_\ell}\\
    = &~ \ell H_\ell + m(1-p_{\ell-1})H_{\ell-1} + \frac{H_{\ell-1}}{p_{\ell-1}}(1-q_\ell^m) ~.
\end{align*}
After all recursive hats arrive, the two brim edges are feasible only if $s,t$ are still disconnected. So with probability at most $q_\ell^m$ the algorithm can obtain additional expected value $H_\ell$ from the brim edges. Hence the total expected value is at most
\[
    \ell H_\ell + m(1-p_{\ell-1})H_{\ell-1} + \frac{H_{\ell-1}}{p_{\ell-1}} + \left(H_\ell-\frac{H_{\ell-1}}{p_{\ell-1}}\right)q_\ell^m ~.
\]
Since $H_\ell-\frac{H_{\ell-1}}{p_{\ell-1}} = n^{2^\ell - 1} - n^{2^{\ell-1}-1}/n^{-2^{\ell-1}/3} > 0$, this expression is convex in $m$, so the maximum over $0\leq m\leq N_\ell$ is attained at either $m = 0$ or $m = N_\ell$. At $m = 0$, the value is exactly $(\ell+1)H_\ell$; at $m = N_\ell$, the value is at most
\begin{align*}
&~ \ell H_\ell + (1-p_{\ell-1})N_\ell H_{\ell-1} + \frac{H_{\ell-1}}{p_{\ell-1}} + H_\ell (1-p_\ell)^{N_\ell}\\
\to &~ (\ell+1)H_\ell-p_{\ell-1}H_\ell+\frac{H_{\ell}}{N_\ell p_{\ell-1}}+H_\ell\exp(-N_\ell p_\ell)\\
= &~ (\ell+1)H_\ell - H_\ell(p_{\ell-1} - p_{\ell-1}^2 - \exp(-p_{\ell-1}^{-1}))\\
< &~ (\ell+1)H_\ell ~,
\end{align*}
the last inequality holds as long as $n$ is sufficiently large. 

Therefore we conclude that on $G_\ell(s,t)$ any non-adaptive algorithm has expected value at most $\alg_\ell=(\ell+1)H_\ell$. Further, in order to obtain value larger than $\alg'_\ell$, the algorithm must either use at least one risky hat, or set a threshold less than $\infty$ for the brim edges, and both strategies connect $s,t$ with probability at least $p_\ell$.

Therefore, for every integer $k\geq 1$, setting $\ell=k$ and taking $n$ sufficiently large gives an instance $G_k(s,t)$ such that any non-adaptive algorithm has competitive ratio at least
\[
\frac{\opt_k}{\alg_k}\geq \frac{(3k+2)H_k-o(H_k)}{(k+1)H_k} \to 3-\frac{1}{k+1} ~. \qedhere
\]

\end{proof}

\subsection{An attempt for simple graphs}
\label{sec:appendix_graphic_simple_fail}

The recursive construction above crucially relies on parallel edges. A natural attempt to adapt it to simple graphs is to replace each pair of brim edges at each level by a triangle: two deterministic edges $(s,v),(v,t)$ of value $H_\ell$, and one edge $(s,t)$ that takes value $H_\ell/p_\ell$ with probability $p_\ell$. 

Such a construction can be analyzed using the same proof strategy as for general graphs. The details are omitted. Eventually we obtain (under the same notation as in the proof of Theorem~\ref{thm:graphic_lower_bound_restate}):
\[
    \opt_\ell=(4\ell+3)H_\ell-o(H_\ell), \qquad  \alg'_\ell=(2\ell+1)H_\ell, \qquad  \alg_\ell=(2\ell+2)H_\ell ~.
\]
Therefore it only gives a lower bound of $\frac{\opt_k}{\alg_k} \to 2-\frac{1}{2k+2}$. Thus, it still remains unclear whether there is a strict separation between adaptive and non-adaptive algorithms for simple graphic matroids.

\subsection{Improved upper bounds for graphic matroids}
\label{sec:appendix_graphic_upper}

\begin{proofof}{Theorem~\ref{thm:graphic_upper_bound_8}}
    The algorithm only accepts edges with one endpoint in $A$ and the other endpoint in $B$. We charge its value to the endpoint in $B$.

    From the perspective of each $v\in B$, each incident edge is \emph{active} with probability $1/2$ (if the other endpoint belongs to $A$),\ and otherwise \emph{inactive} (the threshold is set to $+\infty$ and is equivalent to not arriving). Let $\tilde E(v)$ denote the set of active incident edges, then we set $T_v$ to be the standard threshold that guarantees a $2$-approximation to the prophet inequality instance induced by $\tilde E(v)$. Under the graphic matroid constraint, it is always feasible to accept at least one edge in $E(v)$. Note that it may be possible to accept more than one edge in $E(v)$, but we only count the first one, which is enough for a $2$-approximation. Therefore, by Fact~\ref{fact:random_subset}, the expected value charged to $v$ is
    \[
    \Pr[v\in B]\cdot \frac{1}{2}\cdot \E_{\tilde E(v)}\left[\E\left[\max_{e\in \tilde E(v)}X_e\right]\right] \geq \frac{1}{2}\cdot \frac{1}{2}\cdot \frac{1}{2}\cdot \E\left[\max_{e\ni v}X_e\right] = \frac{1}{8}\cdot \E\left[\max_{e\ni v}X_e\right] ~.
    \]
    Since the value of each accepted edge is charged to exactly one vertex, we have
    \[
    \E[\alg] \geq \frac{1}{8}\sum_{v\in V}\E\left[\max_{e\ni v}X_e\right] \geq \frac{1}{8}\opt ~. \qedhere
    \]
\end{proofof}

For simple graphs, suppose we partition the graph such that $A$ includes each vertex independently with probability $p$, then the analysis of each $v\in B$ boils down to the following lemma.

\begin{lemma}
\label{lem:subsample_prophet inequality}
    Fix a prophet inequality instance with independent variables $X_1,\ldots,X_n$, suppose each variable actually arrives independently with probability $p$, then there is an algorithm that obtains expected reward at least a $\frac{p}{1+p}$ fraction of the prophet's value on the full instance. Further, this ratio is tight.
\end{lemma}

\begin{proof}
    Fix a prophet inequality instance with independent nonnegative random variables $X_1,\ldots,X_n$. We say that $i$ is active if it is included in the sub-instance. 
    
    Following the standard analysis of single-choice prophet inequalities, we decompose the value accepted by the algorithm into $T$ and above $T$:
    \begin{align*}
    \E[\alg] = &~ \sum_{i=1}^n \Pr[\text{algorithm does not stop before }i]\cdot p\cdot\E[(X_i-T)^+]\\
    &~ + \Pr[\text{algorithm accepts some item}]\cdot T ~.
    \end{align*}
    For each $i\in[n]$, let $p_i=\Pr[X_i<T]$. The probability that the algorithm
    does not accept any of the $n$ items is 
    \[
    Q_p(T)=\prod_{i=1}^n (1-p+p p_i) ~.
    \]
    Then for each $i\in[n]$, we can lower bound the probability that the algorithm does not stop before $i$ by the probability that the algorithm does not stop at all, which is exactly $Q_p(T)$. Therefore, we conclude that
    \[
    \E[\alg] \geq Q_p(T) \cdot p \cdot \sum_{i=1}^n \E[(X_i-T)^+] + (1-Q_p(T))\cdot T ~.
    \]
    To balance the two terms, we set $T$ such that $Q_p(T)=\frac{1}{1+p}$, then
    \[
    \E[\alg] \geq \frac{p}{1+p}\cdot \sum_{i=1}^n \E[(X_i-T)^+] + \frac{p}{1+p}\cdot T \geq \frac{p}{1+p}\cdot \E\left[\max_{i\in[n]} X_i\right] ~.
    \]
    Setting such a $T$ is always possible: under the assumption that the distributions are continuous, each $p_i=\Pr[X_i<T]$ is continuous and monotone non-decreasing in $T$, therefore $Q_p(T)$ is a continuous and monotone non-decreasing function of $T$. Thus, since
    \[
    Q_p(0)=(1-p)^n\le 1/(1+p)\leq Q_p(+\infty)=1 ~,
    \]
    there must exist some $T$ such that $Q_p(T)=\frac{1}{1+p}$.

    Next we show that this is tight: fix $p$, let $k$ be a sufficiently large number such that $k/p$ is an integer, and $\eps\to 0$, consider an instance where
    \[
    X_1 = k, \qquad X_2,\ldots,X_{k/p+1} = \begin{cases}
        1/\eps, \qquad \text{with probability }\eps,\\
        0,\qquad \text{otherwise}.
    \end{cases}
    \]
    Then the prophet gets expected value 
    \[
    \E[\max_i X_i] = k(1-\eps)^{k/p}+\frac{1}{\eps}\left(1-(1-\eps)^{k/p}\right) \to k+\frac{k}{p} ~,
    \]
    while any online algorithm can get at most $k+o(k)$, since when $k/p$ is sufficiently large, the number of variables among $X_2,\ldots,X_{k/p+1}$ that survive the sub-sampling will concentrate at $k\pm o(k)$ almost surely. Therefore the ratio tends to $\frac{k}{k+k/p} = \frac{p}{1+p}$.
\end{proof}

We plug the algorithm from Lemma~\ref{lem:subsample_prophet inequality} into Algorithm~\ref{alg:alg1} and optimize the parameter $p$ to obtain an improved upper bound for simple graphs. 

\begin{proofof}{Theorem~\ref{thm:graphic_upper_bound_simple}}
    We slightly modify Algorithm~\ref{alg:alg1}: partition the graph such that each vertex belongs to $A$ independently with probability $p$, and use Lemma~\ref{lem:subsample_prophet inequality} to compute the thresholds for each $v\in B$. Then from the perspective of each $v\in B$, each incident edge is considered independently with probability $p$, so by Lemma~\ref{lem:subsample_prophet inequality} the expected value charged to $v$ is at least 
    \[
    \Pr[v\in B]\cdot \frac{p}{1+p}\E\left[\max_{e\ni v}X_e\right] = \frac{p(1-p)}{1+p}\E\left[\max_{e\ni v}X_e\right] ~. 
    \]
    This gives a competitive ratio of $\frac{1+p}{p(1-p)}$. The optimal choice of $p$ is $p^*=\sqrt 2 - 1$ and gives competitive ratio $\frac{1+p^*}{p^*(1-p^*)}=3+2\sqrt 2 \approx 5.828$.
\end{proofof}

\end{document}